\documentclass[a4paper]{article}
\title{Forecasting on the Accuracy–Timeliness Frontier: Two Novel `Look Ahead' Predictors }
\author{Marc Wildi}

\usepackage{a4wide}
\usepackage{amsmath}
\usepackage{amsfonts} 
\usepackage{amsthm}

\usepackage[utf8]{inputenc}
\usepackage{hyperref}
\usepackage{float}
\newtheorem{Proposition}{Proposition}
\newtheorem{Corollary}{Corollary}
\newtheorem{Theorem}{Theorem}
\DeclareMathOperator{\sign}{sign}
\DeclareMathOperator{\arctantwo}{arctan2}

\usepackage{Sweave}
\begin{document}

\maketitle

\begin{abstract}
\noindent

We re-examine the traditional Mean-Squared Error (MSE) forecasting paradigm by formally integrating an accuracy–timeliness trade-off: accuracy is defined by MSE (or target correlation) and timeliness by advancement (or phase excess).  While MSE-optimized predictors are accurate in tracking levels, they sacrifice dynamic lead, causing them to lag behind changing targets. To address this, we introduce two `look-ahead' frameworks—Decoupling-from-Present (DFP) and Peak-Correlation-Shifting (PCS)—and provide closed-form solutions for their optimization.  Notably, the classical MSE predictor is shown to be a special case within these frameworks. Dually, our methods achieve maximum advancement for any given accuracy level, so our approach reveals the complete efficient frontier of the accuracy–timeliness trade-off, whereas MSE represents only a single point. We also derive a universal upper bound on lead over MSE for any linear predictor under a consistency constraint and prove that our methods hit this ceiling. We validate this approach through applications in forecasting and real-time signal extraction, introducing a leading-indicator criterion and tailored linear benchmarks.


\end{abstract}

\section{Introduction}

Petropoulos et al. (2022), in their comprehensive treatment of forecasting theory and practice, contend that “the theory of forecasting appears mature today … the fact that forecasting is mature does not mean that all has been done.” Building on this perspective, we propose new directions by formalizing a novel forecast accuracy–timeliness dilemma. \\

Forecasting involves several partly competing objectives: accuracy (predicting future levels), timeliness (lead–lag behavior relative to a benchmark, i.e., retardation or advancement), and smoothness (suppressing spurious high-frequency noise). In this sense, forecasting and real-time signal extraction (nowcasting) rest on common methodological foundations, as our examples illustrate.\\

Ideally, one would optimize accuracy, smoothness, and timeliness (AST) in a single objective, producing predictors/nowcasts that track the latent level, detect turning points without systematic delay, and avoid spurious high-frequency noise. The AST setting, however, implies a formal trilemma: improving any one component necessarily worsens at least one of the others. Wildi (2005) and Wildi and McElroy (2019) propose a framework based on this trilemma, but their criterion lacks a closed-form solution. Wildi (2024), (2026a), and (2026b) further study a two-way accuracy–smoothness trade-off. Here, we instead focus on the explicit accuracy–timeliness trade-off in prediction.\\

We introduce two forecasting frameworks—Decoupling-From-Present (DFP) and Peak-Correlation-Shifting (PCS)—and derive exact closed-form solutions to their respective optimization problems.  Both generalize the classical minimum-MSE predictor, recovered as a special case. We prove a duality result: for any target accuracy, our predictors achieve the maximum possible phase excess (lead) relative to the MSE benchmark, yielding an efficient frontier that quantifies the accuracy–timeliness trade-off. While MSE is a single point on this frontier, DFP and PCS trace the full curve. We derive an upper bound on the maximum `meaningful' lead that can be achieved by a linear predictor under a consistency constraint; our predictors attain it, restricting hyperparameters to a valid admissible region.\\

Our method deliberately trades accuracy at selected horizons for a systematic time lead over the MSE-optimal predictor. Unlike classic forecast combinations that pool common information across competing forecasts at a fixed horizon (e.g., Hsiao and Wan, 2014; Makridakis et al., 2018), we exploit dissimilarities in predictors across different horizons. Counterintuitively, decoupling the predictor from the latest  observation—usually considered the most influential—can increase the lead, and overweighting older data can extend it further. In some scenarios where timeliness is paramount, the conventional goal of minimizing forecast error can be inverted to maximize it within a controlled bound.  This demonstrates how prioritizing speed can challenge conventional forecasting wisdom.   We focus on stationary univariate results for brevity, but the approach extends to nonstationary and multivariate cases.\\

We begin in Section \ref{onemulti} by documenting the limitations of the traditional MSE forecasting model. Section \ref{peak_cor_decoup} then introduces the DFP and PCS predictors, deriving closed-form solutions and a dual formulation that yields a new accuracy–timeliness efficient frontier. Lead time is defined in Section \ref{time_shift}, where we relate it to the key hyperparameters and derive, under a consistency constraint, an upper bound on admissible lead. Applications to time-series forecasting and signal extraction—including the construction of leading indicators and the customization of linear predictors—are presented in Section \ref{examples}. Section \ref{concl} then summarizes the main conclusions.

\section{Limitations of Classic Forecast Approach: a Case Study}\label{onemulti}

For illustration we consider forecasting an MA($q$) process 
\[x_t=\mu+\gamma_0\epsilon_t+\gamma_1\epsilon_{t-1}+...+\gamma_q\epsilon_{t-q},\]
where $\gamma_0=1$. To simplify exposition, we assume that $\mu=0$, that $\gamma_1,...,\gamma_q$ are known and that the white noise innovations $\epsilon_t$ are observed.\\ 

The MSE $h$-step ahead forecast $\hat{x}_{hT}^{MSE}$ of $x_{T+h}$ at the sample end $t=T$ is
\begin{eqnarray*}
\hat{x}_{hT}^{MSE}=\left\{\begin{array}{cc}\sum_{k=0}^{q-h} \gamma_{h+k}\epsilon_{T-k}&1\leq h\leq q\\0&\textrm{otherwise.}
\end{array}\right.
\end{eqnarray*}
For $h\leq q$ and $t\in\{q-h+1, \ldots ,T\}$, consider the \emph{forecast filter} $\hat{x}_{ht}^{MSE}=\sum_{k=0}^{q-h} \gamma_{k+h}\epsilon_{t-k}$ with weights $\gamma_{h}, \ldots ,\gamma_{q}$. The cross correlation function (CCF) between $x_{t+\delta}$ and $\hat{x}_{ht}^{MSE}$ at lead $\delta$ ($q\geq \delta>0$) or at lag $\delta$ ($h-q\leq \delta\leq 0$) is given by
\begin{eqnarray*}
\rho(x_{t+\delta},\hat{x}_{ht}^{MSE})&=&\frac{\sum_{k=\max(0,-\delta)}^{\min(q-h,q-\delta)} \gamma_{k+\delta}\gamma_{k+h}}{\sqrt{\sum_{k=0}^q \gamma_k^2\sum_{k=0}^{q-h}\gamma_{k+h}^2}}.
\end{eqnarray*}

For illustration, Fig. \ref{cor} presents the CCF for an exponentially weighted MA(9) process, $x_t=\sum_{k=0}^{9}0.9^k\epsilon_{t-k}$, together with its optimal $h=5$ step-ahead forecast, $\hat{x}_{ht}^{MSE}=\hat{x}_{5t}^{MSE}=\sum_{k=0}^{4}0.9^{k+5}\epsilon_{t-k}$, evaluated at leads and lags $\delta\in\{-4,-3,...,9\}$ (top panel).\\

By design, $\hat{x}_{5t}^{MSE}$ maximizes the CCF at the designated forecast horizon $h = 5$,  as indicated by the kink at the green vertical line. The correlation also increases as the lead shrinks ($\delta \to 0$, with $0\le \delta \le 5$), peaking at 0.86 when $\delta=0$ (black vertical line).
The lower panel plots $x_{t}$ (black line) and its $h=5$-step-ahead predictor, $\hat{x}_{5t}^{MSE}$ (green), for $t = 1, …, 100$. Instead of clearly leading the series, the predictor is largely time-aligned with the process it is meant to forecast. This is at odds with the intent of a truly forward-looking predictor and stems from the tight coupling between $\hat{x}_{ht}^{MSE}$ and $x_t$ implied by MSE optimality, as the CCF makes clear.  Increasing 
$h$ does not weaken this dependence, so changing the forecast horizon does not fix the problem. Although deliberately simple, the example reflects common challenges in economic forecasting, including real-time signal extraction and trend nowcasting (see Section \ref{RTBCA}). Our new methods extend this framework to general forecasting problems and address the issue by explicitly incorporating a left shift (an effective lead) in forecasts for general stationary processes (extensions to nonstationary integrated processes are omitted for brevity).\\

\begin{figure}[H]\begin{center}\includegraphics[height=3in, width=4in]{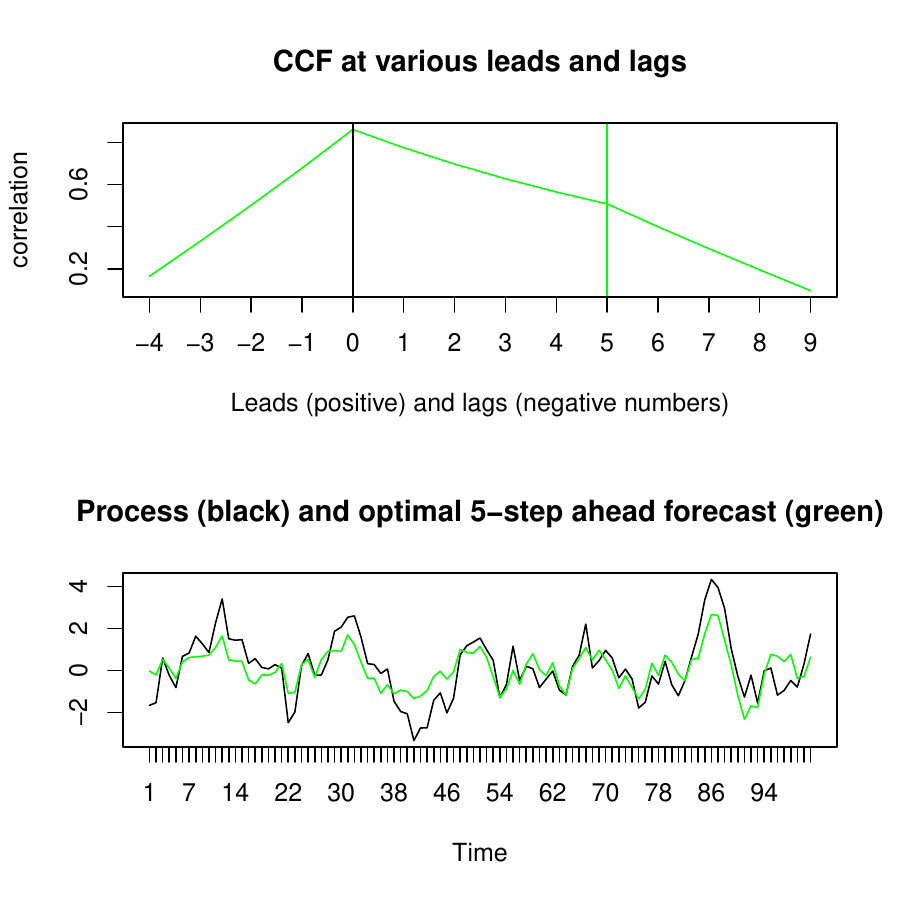}\caption{Top graph: CCF of MSE 5-step ahead forecast filter and  MA(9) process at leads $9\geq\delta> 0$ and lags $-4\leq \delta\leq 0$. Bottom graph: a comparison of realizations of the process (black line) and of the optimal 5-step ahead forecast filter (green line)\label{cor}}\end{center}\end{figure}

We now move beyond the MA(9) example and consider a stationary linear process
\begin{eqnarray}\label{proc}
x_t=\sum_{k=-\infty}^{\infty}\gamma_k\epsilon_{t-k},
\end{eqnarray}
where the impulse response sequence  $\gamma_k$ is square-summable, $\sum_{k=-\infty}^{\infty}\gamma_k^2<\infty$. We assume iid innovations $\epsilon_t$, though the results extend to uncorrelated white noise when focusing on best \emph{linear} prediction. For simplicity, we take $\epsilon_t$ to be observed\footnote{We also analyze AR inversions, but the MA impulse-response representation is generally more informative about predictor dynamics.}; otherwise, $\epsilon_t$ can be recovered using standard inversion methods (see Section \ref{AR3} for an example). Allowing $k < 0$ in \eqref{proc} (i.e., inclusion of future innovations $\epsilon_{t-k}$) accommodates signal-extraction settings with two-sided bi-infinite filters (cf. Section \ref{RTBCA}). Consider a generic (not necessarily MSE-optimal) predictor for $x_{t+h}$ generated by a causal finite-length filter of order $L$ with coefficient vector $\mathbf{b}=(b_0, …, b_{L-1})'$:
\[
\hat{x}_{ht} = \mathbf{b}'\boldsymbol{\epsilon}_t=\sum_{k=0}^{L-1} b_k\epsilon_{t-k},
\]
where $\mathbf{b}=\mathbf{b}(h)$ generally depends on the forecast horizon $h$. Define the length-$L$ coefficient vector for horizon $h$ as $\boldsymbol{\gamma}_h:=(\gamma_h,\gamma_{h+1},...,\gamma_{h+L-1})'$, which yields the (length-$L$) MSE-optimal linear predictor
\[
\hat{x}_{ht}^{MSE}=\sum_{k=0}^{L-1}\gamma_{h+k}\epsilon_{t-k}.
\]
For brevity, we refer to the`MSE predictor at horizon $h$' interchangeably as the weight vector $\boldsymbol{\gamma}_h$ or its implied predictor $\hat{x}_{ht}^{MSE}$. When $h=0$, this reduces to the nowcast, and $\boldsymbol{\gamma}_0$ denotes the contemporaneous coefficient vector. In a stationary ARMA($p,q$) setting, the nowcast coincides with a truncated MA inversion (Wold decomposition, assuming no deterministic component).  As mentioned, we restrict attention to univariate stationary processes, noting that the approaches extend to integrated and multivariate settings (not shown).


\section{Look-Ahead Forecast Optimization Criteria}\label{peak_cor_decoup}

We develop two variants—Decoupling from Present (DFP) and Peak Correlation Shifting (PCS)—designed to advance the predictor relative to the benchmark MSE design (i.e., induce a controllable lead via a leftward shift). 

\subsection{Decoupling from Present }\label{dfp}

As established in Section \ref{onemulti}, the classical MSE predictor in this setting achieves its CCF maximum at the contemporaneous lag $\delta=0$, yielding a coincident rather than a leading forecast. This property is essentially invariant to increases in the forecast horizon $h\leq q$. To induce a systematic lead, we therefore seek to decouple the predictor from the present and formulate the following decoupling-from-present (DFP) problem:
\begin{eqnarray}
&&\max_{\mathbf{b}} \boldsymbol{\gamma}_h'\mathbf{b}\label{dp}\\
\textrm{s.t.}&&\boldsymbol{\gamma}_{0}'\mathbf{b}/\|\boldsymbol{\gamma}_0\|=\alpha_0\nonumber\\
&&\mathbf{b}'\mathbf{b}=1\nonumber,
\end{eqnarray}
where $\|\boldsymbol{\gamma}_0\|=\sqrt{\boldsymbol{\gamma}_{0}'\boldsymbol{\gamma}_{0}}$.
The unit-norm constraint, $\mathbf{b}'\mathbf{b}=1$, implies that the linear objective $\boldsymbol{\gamma}_h'\mathbf{b}$ is proportional to the correlation $\rho(x_{t+h},\hat{x}_{ht})$ at $\delta=h$.
This is computationally advantageous: $\rho(x_{t+h}, \hat{x}_{ht})$ is a nonlinear function of $\mathbf{b}$, whereas the objective in \eqref{dp} is linear, yielding a tractable optimization while preserving the correlation-based interpretation at the target horizon. Likewise, under the unit-length constraint $\alpha_0$ specifies the contemporaneous correlation  $\rho(x_{t},\hat{x}_{ht})$ at $\delta=0$. Setting  $\alpha_0=0$ achieves complete decoupling of the predictor from the present; see Section \ref{la} for an example. Observe that $\alpha_0=\cos(\theta_{0b})$, where $\theta_{0b}$ denotes the angle between $\boldsymbol{\gamma}_{0}$ and $\mathbf{b}$. Hence, $\alpha_0$ parametrizes the phase $\theta_{0b}$ between the predictor $\mathbf{b}$ and the nowcast $\boldsymbol{\gamma}_{0}$ at $\delta=0$, see Fig.\ref{dfp_geometry}. Conversely, $\theta_{0b}$ is determined by $\alpha_0$ up to its sign. Typically, only one of the two solutions ($\pm \theta_{0b}$) is genuinely forward-looking; the alternative is backward-looking and entails a corresponding lag. 
Conceptually, this connects to Wildi and McElroy (2019), who frame a trilemma in terms of time-shift and amplitude functions. In what follows, we assume $|\alpha_0|\le 1$ (feasibility) and that the implied solution delivers a strictly positive objective value (strict positivity), as is typical in applications (the criterion can vanish, for example, for an MA($q$) process when $h>q$). \\ 

\textbf{Remark}: Decoupling the predictor from the present may seem counterintuitive, as the most recent observation $x_T$ at the sample end $t = T$ is typically regarded as highly informative for forecasting. However, decoupling does not imply that $x_T$ is irrelevant. Rather, it acknowledges that achieving an effective lead of the predictor requires a dynamic pattern that differs from the contemporaneous behavior of $x_T$. In this sense, a degree of decoupling—i.e.,  loosening the tie to the present—is intrinsic to constructing a genuinely forward-looking predictor.\\  

The following Theorem provides the closed-form solution to the DFP criterion.

\begin{Theorem}\label{solution_dfp}
Assume $\boldsymbol{\gamma}_0$ and $\boldsymbol{\gamma}_h$ ($h>0$) are linearly independent and  $|\alpha_0|< 1$. Then the solution to criterion \eqref{dp} lies in $\textrm{span}\{\boldsymbol{\gamma}_0,\boldsymbol{\gamma}_h\}$:
\begin{eqnarray}\label{dfp_sol}
\mathbf{{b}}=\lambda_1\boldsymbol{\gamma}_h+\lambda_2\boldsymbol{\gamma}_0
\end{eqnarray}
for some scalars $\lambda_1,\lambda_2$. 
If $\boldsymbol{\gamma}_0'\boldsymbol{\gamma}_h\neq 0$, then  
\begin{eqnarray*}
\lambda_1&=&\frac{\alpha_0\|\boldsymbol{\gamma}_0\|-\lambda_2\boldsymbol{\gamma}_0'\boldsymbol{\gamma}_0}{\boldsymbol{\gamma}_0'\boldsymbol{\gamma}_h}
\end{eqnarray*}
and $\lambda_2$ is determined by the quadratic
\begin{eqnarray}\label{quad}
\lambda_2&=&\frac{-b\pm\sqrt{b^2-4ac}}{2a},
\end{eqnarray}
with coefficients 
\begin{eqnarray*}
a&=&\boldsymbol{\gamma}_0'\boldsymbol{\gamma}_0\left(\frac{\boldsymbol{\gamma}_0'\boldsymbol{\gamma}_0\boldsymbol{\gamma}_h'\boldsymbol{\gamma}_h}{\boldsymbol{(\gamma}_0'\boldsymbol{\gamma}_h)^2}-1\right)\\
b&=&2\alpha_0\|\boldsymbol{\gamma}_0\|\left(1-\frac{\boldsymbol{\gamma}_0'\boldsymbol{\gamma}_0\boldsymbol{\gamma}_h'\boldsymbol{\gamma}_h}{(\boldsymbol{\gamma}_0'\boldsymbol{\gamma}_h)^2}\right)\\
c&=&\frac{\alpha_0^2\boldsymbol{\gamma}_0'\boldsymbol{\gamma}_0\boldsymbol{\gamma}_h'\boldsymbol{\gamma}_h}{(\boldsymbol{\gamma}_0'\boldsymbol{\gamma}_h)^2}-1.
\end{eqnarray*}
Among the two candidate solutions for $\lambda_2$, choose the root that maximizes the objective $\boldsymbol{\gamma}_h'\mathbf{b}$ in \eqref{dp}. 
On the other hand, if $\boldsymbol{\gamma}_0'\boldsymbol{\gamma}_h=0$, then 
\begin{eqnarray*}
\lambda_2&=&\frac{\alpha_0}{\|\boldsymbol{\gamma}_0\|}\\
\lambda_1&=&\pm\sqrt{\frac{1-\alpha_0^2}{\boldsymbol{\gamma}_h'\boldsymbol{\gamma}_h}}.
\end{eqnarray*}
where the appropriate sign maximizes the objective $\boldsymbol{\gamma}_h'\mathbf{b}$. 
Finally, if $|\alpha_0|= 1$, then the solution is $\mathbf{{b}}=\sign(\alpha_0)\boldsymbol{\gamma}_0/\|\boldsymbol{\gamma}_0\|$.
\end{Theorem}

\textbf{Proof:} We first assume $|\alpha_0|< 1$. Then, under criterion \eqref{dp}, the feasible set is the intersection of the unit sphere ($\mathbf{b}'\mathbf{b}=1$) and the circular cone defined by the decoupling constraint $\boldsymbol{\gamma}_{0}'\mathbf{b}=\alpha_0\|\boldsymbol{\gamma}_0\|\|\mathbf{b}\|$ (where $\|\mathbf{b}\|=1$ is relaxed), i.e., the cone with axis $\boldsymbol{\gamma}_{0}$ and semi-angle $\theta_{0b}=\arccos(\alpha_0)$. Moreover, maximizing the projection $\boldsymbol{\gamma}_h'\mathbf{b}$ implies that the optimizer—i.e., the admissible unit vectors on the cone—lies in the two-dimensional subspace spanned by the axis $\boldsymbol{\gamma}_0$ of the cone and the  objective direction $\boldsymbol{\gamma}_h$ (see Fig. \ref{dfp_geometry}). Consequently, 
\[
\mathbf{{b}}=\lambda_1\boldsymbol{\gamma}_h+\lambda_2\boldsymbol{\gamma}_0,
\]
where $\lambda_1$ and $\lambda_2$ are determined by the decoupling and unit-norm constraints. Assuming $\boldsymbol{\gamma}_0'\boldsymbol{\gamma}_h\neq 0$, the decoupling constraint implies
\begin{eqnarray}\label{lm1_lm2}
\lambda_1=\frac{\alpha_0\|\boldsymbol{\gamma}_0\|-\lambda_2\boldsymbol{\gamma}_0'\boldsymbol{\gamma}_0}{\boldsymbol{\gamma}_0'\boldsymbol{\gamma}_h}.
\end{eqnarray}
From the unit-length constraint we deduce 
\begin{eqnarray*}
1=\mathbf{b}'\mathbf{b}=(\lambda_1\boldsymbol{\gamma}_h+\lambda_2\boldsymbol{\gamma}_0)'(\lambda_1\boldsymbol{\gamma}_h+\lambda_2\boldsymbol{\gamma}_0).
\end{eqnarray*}
Substituting \eqref{lm1_lm2}  into the unit‑length condition and simplifying yields the quadratic \eqref{quad} in $\lambda_2$; the appropriate branch (sign) is then selected by maximizing the objective function.
\\
On the other hand, if $\boldsymbol{\gamma}_0'\boldsymbol{\gamma}_h= 0$, then the decoupling constraint implies
\begin{eqnarray}\label{lm2_case2}
\lambda_2=\alpha_0\frac{\|\boldsymbol{\gamma}_0\|}{\boldsymbol{\gamma}_0'\boldsymbol{\gamma}_0}=\frac{\alpha_0}{\|\boldsymbol{\gamma}_0\|}
\end{eqnarray}
Inserted into the length constraint, we obtain
\[
1=\mathbf{b}'\mathbf{b}=(\lambda_1\boldsymbol{\gamma}_h+\lambda_2\boldsymbol{\gamma}_0)'(\lambda_1\boldsymbol{\gamma}_h+\lambda_2\boldsymbol{\gamma}_0)=\lambda_1^2\boldsymbol{\gamma}_h'\boldsymbol{\gamma}_h+\lambda_2^2\boldsymbol{\gamma}_0'\boldsymbol{\gamma}_0.
\]
Inserting \eqref{lm2_case2} then implies
\[
\lambda_1=\pm\sqrt{\frac{1-\alpha_0^2}{\boldsymbol{\gamma}_h'\boldsymbol{\gamma}_h}},
\]
where the appropriate sign is selected to satisfy the objective (e.g., maximize alignment with $\boldsymbol{\gamma}_h $). Finally, if $|\alpha_0|= 1$, then the cone degenerates to a ray along $\boldsymbol{\gamma}_0$ and the decoupling constraint together with the unit-length uniquely determine
\[
\mathbf{{b}}=\sign(\alpha_0)\boldsymbol{\gamma}_0/\|\boldsymbol{\gamma}_0\|,
\]
as claimed.\hfill\qed\\

\begin{figure}[H]\begin{center}\includegraphics[height=4in, width=4in]{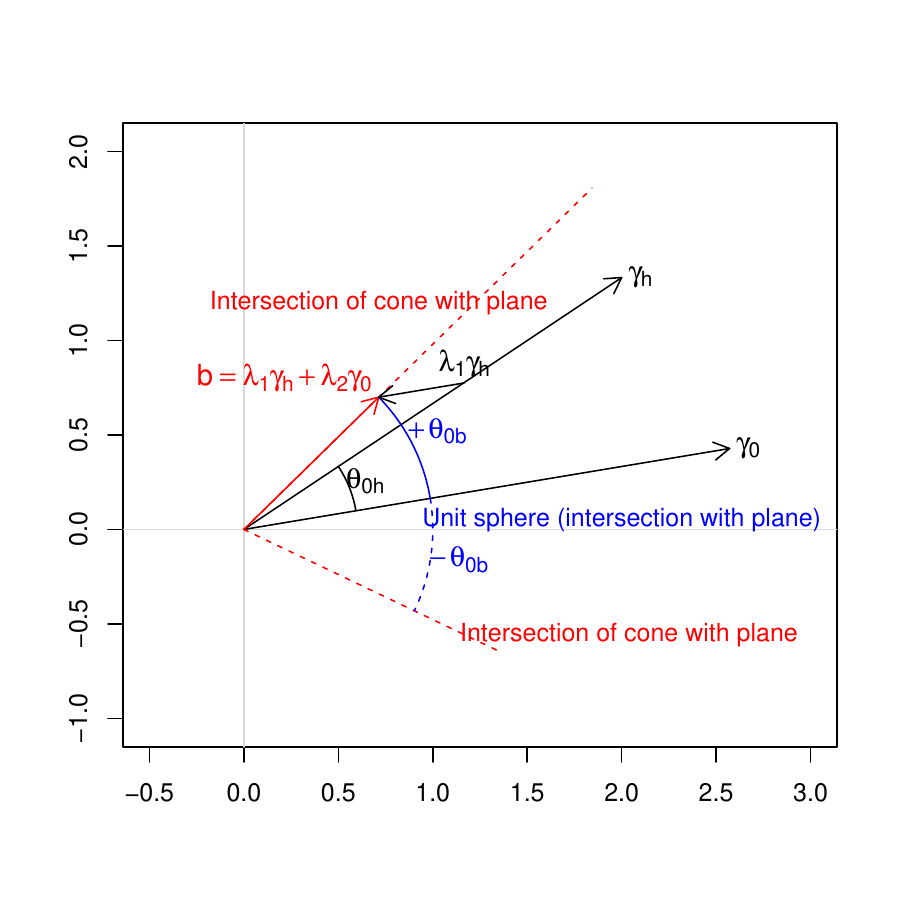}\caption{Geometry of the DFP predictor. The solution $\mathbf{{b}}=\lambda_1\boldsymbol{\gamma}_h+\lambda_2\boldsymbol{\gamma}_0$ lies at the intersection of the plan spanned by $(\boldsymbol{\gamma}_0,\boldsymbol{\gamma}_h)$, the circular cone with semi-angle $\theta_{0b}$ (red), and the unit-sphere (blue). The cone has axis $\boldsymbol{\gamma}_0$ and half-angle $\theta_{0b}=\arccos(\alpha_0)$. The figure depicts one of the two solutions of the quadratic in $\lambda_2$ associated with $+\theta_{0b}$ (with $\lambda_1>0,\lambda_2<0$); the mirrored solution at $-\theta_{0b}$ (with $\lambda_1<0,\lambda_2>0$) does not maximize the objective $\boldsymbol{\gamma}_h'\mathbf{b}$ and corresponds to a lag, so it is omitted. The quantity $\theta_{0b}-\theta_{0h}>0$ represents the phase excess of the DFP relative to the MSE predictor $\boldsymbol{\gamma}_h$ and is associated with its advancement (lead), see Section \ref{time_shift}.\label{dfp_geometry}}\end{center}\end{figure}
\textbf{Remark: } Figure \ref{dfp_geometry} shows that, for the solution based on the root $\lambda_2<0$, the angle $\theta_{0b}$ between the DFP predictor $\mathbf{{b}}$ and the nowcast $\boldsymbol{\gamma}_0$ exceeds the angle $\theta_{0h}$ between $\boldsymbol{\gamma}_h$ and $\boldsymbol{\gamma}_0$. The difference $\theta_{0b}-\theta_{0h}>0$ represents the excess phase induced by decoupling and is directly related to the predictor’s lead under a `standard' configuration (see Section \ref{time_shift}). The alternative root $\lambda_2>0$, corresponding to the branch at  $-\theta_{0b}$, produces a phase-reversed solution and thus corresponds to a lag (i.e., a negative lead). Complete decoupling occurs when $\theta_{0b}=\pi/2$, in which case $\mathbf{{b}}$ is orthogonal to the nowcast $\boldsymbol{\gamma}_0$. \\

We can derive the distribution of the DFP predictor conditional on the distribution of an estimate $\boldsymbol{\hat{\gamma}}_0$ of the nowcast  $\boldsymbol{\gamma}_0$. To this end, we approximate Equation \eqref{dfp_sol} with
\begin{eqnarray}\label{dfp_gamma0}
\mathbf{b}\approx(\lambda_1 \mathbf{F}^h+\lambda_2\mathbf{I})\boldsymbol{\gamma}_0,
\end{eqnarray}
where $\mathbf{F}$ is the $L\times L$ forward operator  and $\mathbf{I}$ is the identity matrix.  The matrix 
$\mathbf{F}$ has ones on its first superdiagonal and zeroes elsewhere. Note that the last $h$ entries of $\mathbf{F}^h\boldsymbol{\gamma}_0$ are zero, so the approximation in \eqref{dfp_gamma0} is valid provided $\gamma_{k}\approx 0$ for $L-1\geq k\geq L-1-h$, which holds for sufficiently large $L$ under stationarity. 

\begin{Corollary}\label{dfp_dist}
Let $\hat{\boldsymbol{\gamma}}_0$ be an estimator of $\boldsymbol{\gamma}_0$ with mean $\boldsymbol{\mu}_{\gamma_0}$ and covariance matrix $\boldsymbol{\Sigma}_{\gamma_0}$ and assume $\lambda_1,\lambda_2$ are fixed. Then the DFP predictor in Equation \eqref{dfp_sol} has mean and covariance
\begin{eqnarray}\label{ci_dfp}
\boldsymbol{\mu}_b\approx(\lambda_1 \mathbf{F}^h+\lambda_2\mathbf{I})\boldsymbol{\mu}_{\gamma_0}~and~\boldsymbol{\Sigma}_b\approx(\lambda_1 \mathbf{F}^h+\lambda_2\mathbf{I})\boldsymbol{\Sigma}_{\gamma_0}(\lambda_1 \mathbf{F}^h+\lambda_2\mathbf{I})',
\end{eqnarray}
where the approximations are valid for $L$ sufficiently large. Moreover, if $\hat{\boldsymbol{\gamma}}_0$ is (asymptotically) Gaussian, then $\mathbf{\hat{b}}$ is (asymptotically) Gaussian as well.
\end{Corollary}

The result follows immediately from \eqref{dfp_gamma0} and standard properties of linear transformations. For the (asymptotic) distribution of the classic MSE estimate $\hat{\boldsymbol{\gamma}}_0$ of $\boldsymbol{\gamma}_0$, see, for instance, Brockwell and Davis (1993) and Cox (1990)\footnote{For an ARMA-process, the distribution of the MA-inverted weights $\gamma_k$, $k=1,...,L-1$ (noting that $\gamma_0=1$) can be derived from the distribution of the AR- and MA-estimates via the delta method as elaborated in Cox (1990), cf. Appendix \ref{sim_stu}.}. \\

A limitation of the corollary is that it treats $\lambda_1$ and $\lambda_2$ as fixed, even though both depend on $\boldsymbol{\gamma}_0$ through Theorem \ref{solution_dfp}. We address this issue below. \\

For further consideration, we propose an alternative DFP-MSE decoupling criterion given by
\begin{eqnarray}
&&\min_{\mathbf{b}} (\mathbf{b}-\boldsymbol{\gamma}_h)'(\mathbf{b}-\boldsymbol{\gamma}_h)\label{dp2}\\
\textrm{s.t.}&&\boldsymbol{\gamma}_{0}'\mathbf{b}=\alpha_0.\nonumber 
\end{eqnarray}
This objective penalizes the squared deviation of $\mathbf{b}$  from the MSE-predictor $\boldsymbol{\gamma}_h$ and thus minimizes the mean-squared prediction error for $x_{t+h}$ under the DFP predictor. In particular, if $\alpha_0=\boldsymbol{\gamma}_{0}'\boldsymbol{\gamma}_{h}$, the solution recovers $\mathbf{{b}}=\boldsymbol{\gamma}_h$. In this formulation, the unit-norm constraint from the earlier DFP criterion can be omitted, simplifying both the geometry and the optimization, but the hyperparameter $\alpha_0$ loses its intuitive interpretation as a correlation and becomes harder to interpret. The objective captures accuracy, while lead is enforced by the decoupling constraint, thereby formalizing the accuracy–timeliness trade-off in Criterion \eqref{dp2} (see Section \ref{ac_ti_di}). Although minimizing \eqref{dp2} is natural, some non-standard configurations instead require maximizing the MSE (under boundedness) to obtain a sizeable lead from $\mathbf{b}$, see Appendix \eqref{l_ns_s}. 


\begin{Proposition}\label{dfp_mse_prop}
Assume that $\boldsymbol{\gamma}_0$ and $\boldsymbol{\gamma}_h$ ($h>0$) are linearly independent. Then the solution to criterion \eqref{dp2} is 
\begin{eqnarray}\label{mse_dfp}
\mathbf{{b}}=\boldsymbol{\gamma}_h+{\lambda}\boldsymbol{\gamma}_0,
\end{eqnarray}
where the scalar ${\lambda}$ is determined by 
\begin{eqnarray}\label{lambda_mse}
{\lambda}=\frac{\alpha_0-\boldsymbol{\gamma}_h'\boldsymbol{\gamma}_0}{\boldsymbol{\gamma}_0'\boldsymbol{\gamma}_0}.
\end{eqnarray}
\end{Proposition}

\textbf{Proof:} The solution to \eqref{dp2} lies on the affine hyperplane defined by the decoupling constraint $\boldsymbol{\gamma}_{0}'\mathbf{b}=\alpha_0$. By the least-squares optimality principle, it is the orthogonal projection of $\boldsymbol{\gamma}_{h}$ onto this hyperplane. Because 
$\boldsymbol{\gamma}_{0}$ is normal to the hyperplane, the optimizer must have the form
\[
\mathbf{{b}}=\boldsymbol{\gamma}_h+{\lambda}\boldsymbol{\gamma}_0
\]
for some scalar ${\lambda}$, which is determined by enforcing the decoupling constraint:
\[
\boldsymbol{\gamma}_{0}'(\boldsymbol{\gamma}_h+{\lambda}\boldsymbol{\gamma}_0)=\alpha_0,
\]
as claimed. \hfill\qed\\

When $\alpha_0=\boldsymbol{\gamma}_h'\boldsymbol{\gamma}_0$ in \eqref{lambda_mse}, we have $\lambda=0$ and $\mathbf{{b}}=\boldsymbol{\gamma}_h$ reproduces the MSE predictor, as noted earlier. The distribution of the MSE-DFP predictor follows from Corollary \ref{dfp_dist} by setting $\lambda_1=1$ and $\lambda_2=\lambda$, treating $\lambda$ as fixed. We now extend this result to incorporate the randomness of $\lambda$, which arises when $\lambda$ is defined as a function of an estimator $\boldsymbol{\hat{\gamma}}_0$ of $\boldsymbol{\gamma}_0$, which is substituted into \eqref{lambda_mse}.

\begin{Corollary}\label{mse_dfp_dist}
Suppose $\boldsymbol{\gamma}_0$ and $\boldsymbol{\gamma}_h$ ($h>0$) are linearly independent. Let  $\hat{\boldsymbol{\gamma}}_0$ be an estimator of $\boldsymbol{\gamma}_0$ with mean $\boldsymbol{\mu}_{\gamma_0}\to\boldsymbol{\gamma}_0\neq \mathbf{0}$ and covariance matrix $\boldsymbol{\Sigma}_{\gamma_0}$ with $\boldsymbol{\Sigma}_{\gamma_0}\to\mathbf{0}$ for $L\to\infty$. Then, for filter length $L$ and sample size $T$ sufficiently large, the mean $\boldsymbol{\mu}_b$ and covariance matrix $\boldsymbol{\Sigma}_b$ of the DFP-MSE \eqref{mse_dfp}, based on  stochastic $\lambda=\lambda(\hat{\boldsymbol{\gamma}}_0)$ in \eqref{lambda_mse}, can be approximated by
\begin{eqnarray}
\boldsymbol{\mu}_b&\approx&(\mathbf{F}^h+{\tilde{\lambda}}\mathbf{I})\boldsymbol{\gamma}_0\label{gen_dist_mu}\\
\boldsymbol{\Sigma}_{b}&\approx&\mathbf{J}\boldsymbol{\Sigma}_{\gamma_0}\mathbf{J}'\label{gen_dist_sigma}
\end{eqnarray}
 where
\begin{eqnarray}
\tilde{\lambda}&=&\frac{\alpha_0-(\mathbf{F}^h\boldsymbol{\gamma}_0)'\boldsymbol{\gamma}_0}{\boldsymbol{\gamma}_0'\boldsymbol{\gamma}_0}\label{lambda_til}\\
\mathbf{J}&=&\mathbf{F}^h+\boldsymbol{\gamma}_0\partial\boldsymbol{\tilde{\lambda}}'+\tilde{\lambda}\mathbf{I}\label{Jac}\\
{d}{\tilde{\lambda}}_k&=&-\frac{I_{\{k\geq h\}}\gamma_{k-h}+I_{\{k\leq L-1-h\}}\gamma_{k+h}}{\|\boldsymbol{\gamma}_0\|}-\frac{\alpha_0-(\mathbf{F}^h\boldsymbol{\gamma}_0)'\boldsymbol{\gamma}_0}{\|\boldsymbol{\gamma}_0\|^2}2\gamma_k\label{dlambdak}
\end{eqnarray} 
with ${d}{\tilde{\lambda}}_k$ denoting the $k$-th element of the row vector $\partial\boldsymbol{\tilde{\lambda}}'$ and $I_{\{\}}$ the indicator function. If the first weight of $\boldsymbol{\gamma}_0$ is held fixed ($\gamma_0=1$), then the first row and the first column of $\boldsymbol{\Sigma}_{\gamma_0}$ are zero; additionally, in \eqref{Jac} the first column of  $\boldsymbol{\gamma}_0\partial\boldsymbol{\tilde{\lambda}}'+\tilde{\lambda}\mathbf{I}$ is set to zero. 
\end{Corollary}

\textbf{Proof}: When $\lambda=\lambda(\hat{\boldsymbol{\gamma}}_0)$ in the DFP predictor is not treated as fixed, Equation \eqref{gen_dist_mu} follows by continuity and by using the approximation $\boldsymbol{\gamma}_h\approx\mathbf{F}^h\boldsymbol{\gamma}_0$ when  $L$ is sufficiently large. To obtain the covariance matrix in \eqref{gen_dist_sigma}, we apply the delta method (see Cox, 1990) and use the first-order Taylor expansion
\[
\mathbf{b}(\hat{\boldsymbol{\gamma}}_0)\approx \mathbf{b}({\boldsymbol{\gamma}_0})+\mathbf{J}(
\hat{\boldsymbol{\gamma}}_0-{\boldsymbol{\gamma}_0}),
\]
assuming $T$ is large. The Jacobian $\mathbf{J}$ has entries 
\begin{eqnarray}\label{jij}
\mathbf{J}_{ij}=\partial b_{i-1}/\partial\gamma_{j-1}=\mathbf{F}^h_{ij}+\frac{\partial\tilde{\lambda}}{\partial\gamma_{j-1}}\gamma_{i-1}+\tilde{\lambda} \mathbf{I}_{ij},~1\leq i,j\leq L,
\end{eqnarray}
with $\tilde{\lambda}=\tilde{\lambda}(\boldsymbol{\gamma}_0)$ as specified in \eqref{lambda_til} and $\mathbf{I}_{ij}=\delta_{ij}$ the Kronecker symbol. Hence, $\partial\tilde{\lambda}/\partial\gamma_{j-1}=d\tilde{\lambda}_{j}$ as in \eqref{dlambdak}. It follows that 
\[
\boldsymbol{\Sigma}_{b}\approx\mathbf{J}\boldsymbol{\Sigma}_{\gamma_0}\mathbf{J}',
\]
which establishes the claim.  If the first component $\gamma_0$ of $\boldsymbol{\gamma}_0$ is held fixed, then the first row and the first column of $\boldsymbol{\Sigma}_{\gamma_0}$ vanish. In addition, the partial derivatives with respect to $\gamma_0$ in \eqref{jij} disappear; equivalently, the first column of $\boldsymbol{\gamma}_0\partial\boldsymbol{\tilde{\lambda}}'+\tilde{\lambda}\mathbf{I}$ is set to zero (the first row does generally not vanish because $\tilde{\lambda}$ depends on $\gamma_k$, $L>k>0$).  \hfill\qed\\

\textbf{Remark}: an extension of Corollary \ref{dfp_dist} to the case of random $\lambda_1,\lambda_2$  can be obtained in the same way. In that setting, however,  $\lambda_i(\boldsymbol{\gamma}_0)$, for $i=1,2$—and therefore $\mathbf{{b}}=\lambda_1\boldsymbol{\gamma}_h+\lambda_2\boldsymbol{\gamma}_0$—become more involved functions of $\boldsymbol{\gamma}_0$; see Theorem \ref{solution_dfp}.\\

Both DFP predictors—those defined by criteria \eqref{dp} and \eqref{dp2}—are linear combinations of $\boldsymbol{\gamma}_h$ and $\boldsymbol{\gamma}_0$, but  entail distinct trade-offs. The mean-squared formulation \eqref{dp2} yields a geometrically simple construction with a unique solution and is especially convenient  for extensions to integrated processes (not pursued here). By contrast,  \eqref{dp} induces a richer geometry and reduces to quadratics with two candidate solutions. Its chief advantage is interpretability: under the unit-norm constraint $\mathbf{b}'\mathbf{b}=1$,  the decoupling parameter has a direct correlation meaning at $\delta=0$, namely $\alpha_0=\rho(x_t,\hat{x}_{ht})$,  which facilitates hyperparameter selection. Criterion \eqref{dp2} can be adapted to endow its hyperparameter with a correlation interpretation as well, albeit at the cost of a more complex solution structure.\\

Let $\theta_{0b}$ and $\theta_{0h}$ represent the angles formed between $\boldsymbol{\gamma}_{0}$ and $\mathbf{{b}}$, and between $\boldsymbol{\gamma}_{0}$ and $\boldsymbol{\gamma}_{h}$, respectively, see Fig.\ref{dfp_geometry_2}. For illustration we assume the nowcast and the MSE predictor to lie in the first quadrant, so that $\theta_{0h}<\pi/2$ (the general case is addressed in Section \ref{time_shift}).  When $\lambda < 0$, a phase excess occurs where $\theta_{0b} > \theta_{0h}$. Intuitively,  the MSE predictor’s lead over the nowcast is tied to  $\theta_{0h}$
 , so doing more of the same—i.e., taking $\theta_{0b}>\theta_{0h}$—further increases $\mathbf{b}$'s lead; we formalize this later, noting that this intuition does not always hold. Strict positivity $\boldsymbol{\gamma}_h'\mathbf{b}>0$ implies $\beta=\theta_{hb}<\pi/2$. In some circumstances, when $\theta_{0h}<\pi/2$ (as assumed), it is meaningful to formulate a more stringent positivity rule 
\begin{eqnarray}\label{str_pos}
\theta_{hb}<\pi/2-\theta_{0h}
\end{eqnarray}
for the phase excess. This guarantees that the DFP predictor correlates positively with \emph{both} the $h$-step ahead MSE predictor and the nowcast (or the original process $x_t$).\\

The Appendix states formal results that relate these angles to  $\lambda$ and to the hyperparameter $\alpha_0$ (Proposition \ref{proplambda0} and Corollary \ref{cor_al}). Unlike Fig.\ref{dfp_geometry}, Fig.\ref{dfp_geometry_2} emphasizes the so-called DFP-triangle with vertices $\mathbf{0}$, $\boldsymbol{\gamma}_h$ and $\mathbf{b}$; this configuration will be useful when connecting $\lambda$ to an effective measure of the DFP predictor’s lead (or advancement); cf. Section \ref{time_shift}. \\

\begin{figure}[H]\begin{center}\includegraphics[height=4in, width=4in]{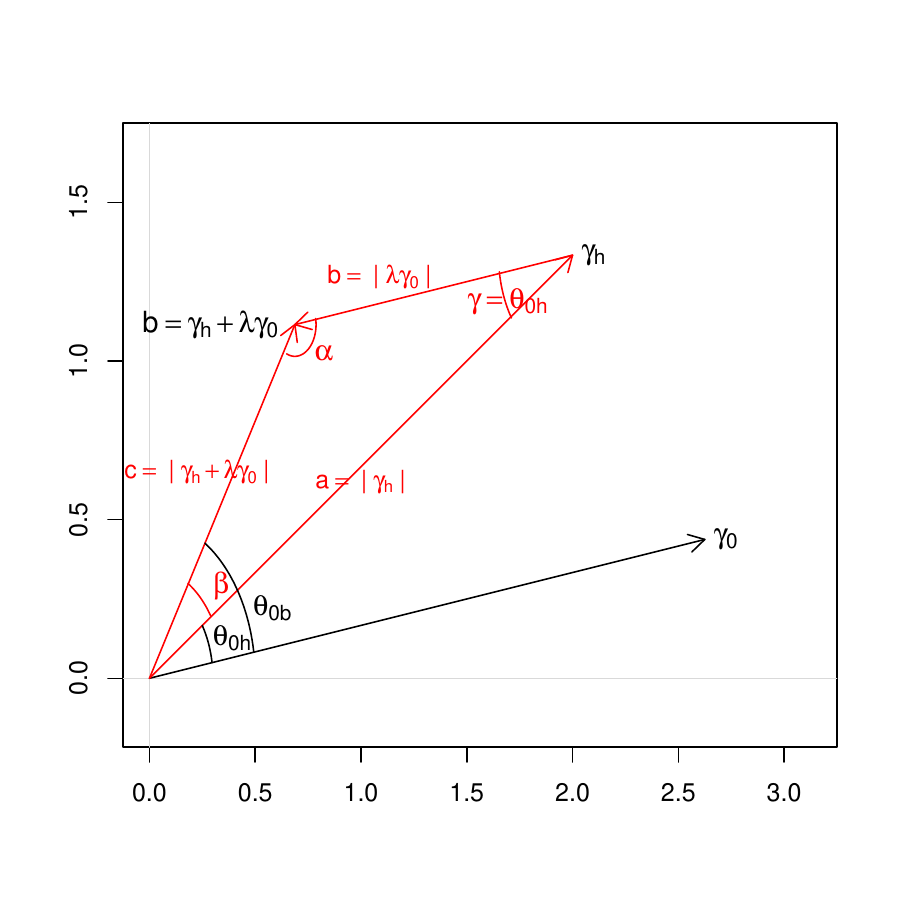}\caption{Side lengths $a$, $b$ and $c$ (red) as well as angles $\alpha$, $\beta$ and $\gamma$ of the DFP-triangle spanned by the vectors $\boldsymbol{\gamma}_h$ and $\mathbf{{b}}=\boldsymbol{\gamma}_h+\lambda\boldsymbol{\gamma}_0$ (black), with $\lambda<0$.\label{dfp_geometry_2}}\end{center}\end{figure}
Complementing the MA form \eqref{mse_dfp}, an AR inversion further clarifies the structure of the DFP principle. For a stationary, invertible zero-mean ARMA process $x_t$, let $\boldsymbol{\phi}$ ($\phi_0=1$) be a finite (length-$L$) AR inversion. For sufficiently large $L$, $\boldsymbol{\phi}'\mathbf{x}_t\approx\epsilon_t$, 
with $\mathbf{x}_t:=(x_t,...,x_{t-(L-1)})'$. The AR-form DFP weights $\mathbf{a}=\mathbf{a}(\lambda)$ follow from convolving $\boldsymbol{\phi}$ with $\mathbf{b}(\lambda)$:
\begin{eqnarray}\label{ainv}
\mathbf{a}(\lambda) = \boldsymbol{\phi} \cdot\mathbf{b}(\lambda)= \boldsymbol{\phi} \cdot \boldsymbol{\gamma}_h + \lambda \boldsymbol{\phi} \cdot \boldsymbol{\gamma}_0 \approx \boldsymbol{\phi} \cdot \boldsymbol{\gamma}_h + \lambda \mathbf{e_1},
\end{eqnarray}
where $\boldsymbol{\phi} \cdot \boldsymbol{\gamma}_h$ is the AR-inverted MSE predictor and $\mathbf{e}_1=(1, 0, \dots, 0)$. Since for large $L$ the AR inversion $\boldsymbol{\phi}$ approximately cancels the MA inversion $\boldsymbol{\gamma}_0$, $\lambda$ affects only the lag-zero weight of the MSE predictor in the AR representation.\\

Equation \eqref{ainv} neatly captures the DFP criterion's dual purpose: it matches the standard MSE predictor by replicating lags $k>0$ in $\boldsymbol{\phi} \cdot \boldsymbol{\gamma}_h$, and enforces decoupling by adjusting only lag zero in  $\lambda \mathbf{e_1}$. The magnitude  $|a_0|$ of the weight on $x_t$ can decrease or increase, as $\lambda<0$ (phase excess) becomes more negative, depending on the sign of the lag-zero weight, i.e.,  $\sign((\boldsymbol{\phi}\cdot\boldsymbol{\gamma}_h)_0)$, of the MSE predictor (assuming $|\lambda|<|(\boldsymbol{\phi}\cdot\boldsymbol{\gamma}_h)_0|$). Although the AR form \eqref{ainv} is simple, the MA impulse response \eqref{mse_dfp} is more informative about the predictor’s dynamics, and the look-ahead predictor in the next section offers no comparably simple AR form; hence the impulse response (MA-) representation $\mathbf{b}(\lambda)$ of the prediction problem remains our focus.\\

To conclude, we briefly comment on the linear-independence assumption for the finite-length predictors $\boldsymbol{\gamma}_0$ and $\boldsymbol{\gamma}_h$ (length $L$). If it fails (i.e., $\boldsymbol{\gamma}_0\propto\boldsymbol{\gamma}_h$), the DFP problem becomes degenerate since the objective is fixed by the constraint. A practical remedy is to slightly perturb one or both of  $\boldsymbol{\gamma}_0,\boldsymbol{\gamma}_h$ to restore linear independence.\footnote{If at least one entry $\gamma_k\neq 0$ for $k=L-h,\ldots,L-1$, then $\mathbf{F}^h\boldsymbol{\gamma}_0$ (approximately $\boldsymbol{\gamma}_h$) and $\boldsymbol{\gamma}_0$ are linearly independent.} These perturbations can be made arbitrarily small, leaving the performance of the underlying MSE predictor (or nowcast) essentially unchanged; we do not elaborate further for brevity. Cases where $\boldsymbol{\gamma}_0\propto\boldsymbol{\gamma}_h$ can arise naturally: when the DGP is an AR(1) process, $\boldsymbol{\gamma}_0$ is proportional to $\boldsymbol{\gamma}_h$ for all $h$; however, proportionality might also appear for periodic or seasonal structures at specific $h$; finally, collinearity could be accidental, without deeper implications for the DGP. In all such cases, the perturbation scheme can be used to induce effective leads with the modified DFP predictors.

\subsection{Peak Correlation Shifting }

We propose an alternative mechanism to attenuate \emph{indirectly} the MSE predictor’s contemporaneous coupling with $\boldsymbol{\gamma}_0$ by relocating the peak of its CCF. As shown in Fig. \ref{cor}, the MSE predictor peaks at  $\delta = 0$.  By shifting this peak to the target forecast horizon $\delta=h$, the predictor becomes genuinely forward-looking. Accordingly, we formulate the Peak Correlation Shifting (PCS) problem: 
\begin{eqnarray}
&&\max_{\mathbf{b}} \boldsymbol{\gamma}_{{h}}'\mathbf{b}\label{peak_cor_sh_crit}\\
&&(\boldsymbol{\gamma}_{h-1}-\boldsymbol{\gamma}_{h})'\mathbf{b}=\beta_h\nonumber\\ 
&&\mathbf{b}'\mathbf{b}=1,\nonumber
\end{eqnarray}
where   $\beta_h$ is a hyperparameter that controls the shape of the CCF. As with the DFP criterion, we assume feasibility—i.e., $|\beta_h|/\|\boldsymbol{\gamma}_{h-1}-\boldsymbol{\gamma}_{h}\|\leq 1$—and positivity.\\

For interpretation, note that in many applications the CCF attains its maximum at a lead smaller than $h$ (i.e., the peak lies to the left of $\delta=h$). As a result—often, though not always—$(\boldsymbol{\gamma}_{h-1}-\boldsymbol{\gamma}_{h})'\mathbf{b}>0$ at $\delta=h$ (cf. Fig.\ref{cor}). The PCS criterion \eqref{peak_cor_sh_crit} directly controls this term, which measures the change in the forecast CCF from lead $h-1$ to $h$. A necessary (though not sufficient) condition to move the peak to $\delta=h$ is $(\boldsymbol{\gamma}_{h-1}-\boldsymbol{\gamma}_{h})'\mathbf{b}<0$, opposite to Fig.\ref{cor}; this can be enforced by choosing $\beta_h$ in \eqref{peak_cor_sh_crit}.  \\
In some settings, however, constraining only the single lead $\delta=h$ may be too weak to generate a meaningful peak shift. A stronger alternative is to impose the constraint over a neighborhood of leads (Appendix \ref{pcs_ext}), but for clarity we proceed with the simple PCS specification above.  \\

The proposed PCS criterion \eqref{peak_cor_sh_crit} is structurally similar to the DFP formulation \eqref{dp}. In particular, the solution to \eqref{peak_cor_sh_crit} follows directly from Theorem \ref{solution_dfp} upon the substitutions
\[
\boldsymbol{\gamma}_0\to\boldsymbol{\gamma}_{h-1} - \boldsymbol{\gamma}_h, \textrm{~~} \alpha_0\to \beta_h/\|\boldsymbol{\gamma}_{h-1}-\boldsymbol{\gamma}_h\|,
\]
assuming $\boldsymbol{\gamma}_{h-1}$ and $\boldsymbol{\gamma}_h$ are linearly independent. Consequently,
\begin{eqnarray}\label{orig_pcs}
\mathbf{b}=\lambda_1\boldsymbol{\gamma}_h+\lambda_2(\boldsymbol{\gamma}_{h-1}-\boldsymbol{\gamma}_h)\approx\left((\lambda_1-\lambda_2)\mathbf{F}^h+\lambda_2\mathbf{F}^{h-1}\right)\boldsymbol{\gamma}_0,
\end{eqnarray}
where the approximation holds for sufficiently large $L$ under stationarity. The distribution of the PCS predictor can be obtained from  Corollary \ref{dfp_dist} (assuming fixed $\lambda_1,\lambda_2$) with the corresponding adjustments. \\

By analogy with \eqref{dp2}, we can define an MSE-PCS criterion as
\begin{eqnarray}\label{pcs_mse}
&&\min_{\mathbf{b}} (\boldsymbol{\gamma}_h-\mathbf{b})'(\boldsymbol{\gamma}_h-\mathbf{b})\\
&&\left(\boldsymbol{\gamma}_{h-1}-\boldsymbol{\gamma}_h\right)'\mathbf{b}=\beta_h\nonumber,
\end{eqnarray}
where we omit the unit-norm constraint $\mathbf{b}'\mathbf{b}=1$. The  solution is 
\begin{eqnarray}\label{pcs_sol}
\mathbf{b}=\boldsymbol{\gamma}_h+\lambda(\boldsymbol{\gamma}_{h-1}-\boldsymbol{\gamma}_h)\approx\left(\mathbf{F}^h+\lambda(\mathbf{F}^{h-1}-\mathbf{F}^h)\right)\boldsymbol{\gamma}_0
\end{eqnarray}
and it follows from Proposition \ref{dfp_mse_prop} with routine adjustments. The distribution of $\mathbf{b}$ is obtained as in Corollary \ref{mse_dfp_dist}, allowing for stochastic $\lambda=\lambda(\hat{\boldsymbol{\gamma}}_0)$.\\

Because the DFP and PCS predictors lie in $\textrm{span}\{\boldsymbol{\gamma}_0,\boldsymbol{\gamma}_h\}$ and $\textrm{span}\{\boldsymbol{\gamma}_{h-1},\boldsymbol{\gamma}_h\}$, respectively, they generally differ unless these subspaces coincide—i.e., unless $\boldsymbol{\gamma}_{h-1}$ is a linear combination of $\boldsymbol{\gamma}_{h}$ and $\boldsymbol{\gamma}_0$.\footnote{For an AR(2) process, the Yule–Walker equations imply this collinearity, so DFP and PCS coincide when $\alpha_0$ and $\beta_h$ match (not shown).} Figure \ref{pcs_geometry} illustrates the geometry of the PCS solution $\mathbf{b}$ in two cases, labeled $\mathbf{b}_1$ and $\mathbf{b}_2$, corresponding to $\beta_h=0$ (red) and $\beta_h<0$ (blue), respectively.  \\
When $\beta_h=0$ , the PCS constraint implies 
\begin{equation}\label{b1_0}
0=\mathbf{b}_1'\boldsymbol{\gamma}_{h-1}- \mathbf{b}_1'\boldsymbol{\gamma}_h=\|\boldsymbol{\gamma}_{h-1}\|\cos(\theta_{hb1}+\theta_{hh-1})-\|\boldsymbol{\gamma}_{h}\|\cos(\theta_{hb1}),
\end{equation} 
where we use $\|\mathbf{b}_1\|=1$, $\theta_{hh-1}$ is the angle between $\boldsymbol{\gamma}_{h-1}$ and $\boldsymbol{\gamma}_{h}$, and $\theta_{hb1}$ is the angle between $\mathbf{b}_1$ and $\boldsymbol{\gamma}_{h}$. \\
If $L$ is large enough and $\gamma_{h-1}\neq 0$, then 
\begin{equation}\label{ineq_pos}
\|\boldsymbol{\gamma}_h\|^2-\|\boldsymbol{\gamma}_{h-1}\|^2=\sum_{k=h}^{h+L-1}\gamma_k^2-\sum_{k=h-1}^{h+L-2}\gamma_k^2=\gamma_{L+h-1}^2-\gamma_{h-1}^2< 0
\end{equation}
since for stationary processes $\gamma_{L+h-1}\to 0$ as $L\to\infty$. Substituting this inequality into \eqref{b1_0} gives
\[
\cos(\theta_{hb1}+\theta_{hh-1})< \cos(\theta_{hb1}),
\]
which implies that, within the plane spanned by the two MSE predictors, $\mathbf{b}_1$ lies on the  side of $\boldsymbol{\gamma}_h$ opposite to  $\boldsymbol{\gamma}_{h-1}$. Equivalently, $\theta_{hb1}>0$ represents a positive phase excess of $\mathbf{b}_1$ relative to $\boldsymbol{\gamma}_h$. \\
The same qualitative conclusion applies when $\beta_h<0$ (the blue case) or when $\gamma_{h-1}=0$ in \eqref{ineq_pos}.\footnote{In this case \eqref{ineq_pos} becomes (asymptotically) an equality as $L\to\infty$, implying $\theta_{hb1}>0$ still holds; strictness follows from \eqref{b1_0} and $\theta_{hh-1}>0$, provided the two MSE predictors are not collinear.} The appendix provides a formal link between the hyperparameter $\beta_h$ and the phase excess $\theta_{hb}$ (Proposition \ref{l_t_b}).  \\

Finally, it is important to note that the AR-representation of the PCS predictor does not share the simple structure of the DFP predictor shown in Eq. \eqref{ainv}, where the influence of $\lambda$ is strictly limited to the lag-zero weight because the convolution term  $\boldsymbol{\phi}\cdot\boldsymbol{\gamma}_0\approx\mathbf{e}_1$ is an  identity. In the PCS predictor \eqref{pcs_sol}, however, the convolution term $\boldsymbol{\phi}\cdot(\boldsymbol{\gamma}_{h-1}-\boldsymbol{\gamma}_h)$ does not resolve to the identity. Consequently, the impact of $\lambda$ is not localized at lag zero; rather, it typically propagates across all lags, both within the MA as well as within the AR representations.

\begin{figure}[H]\begin{center}\includegraphics[height=4in, width=4in]{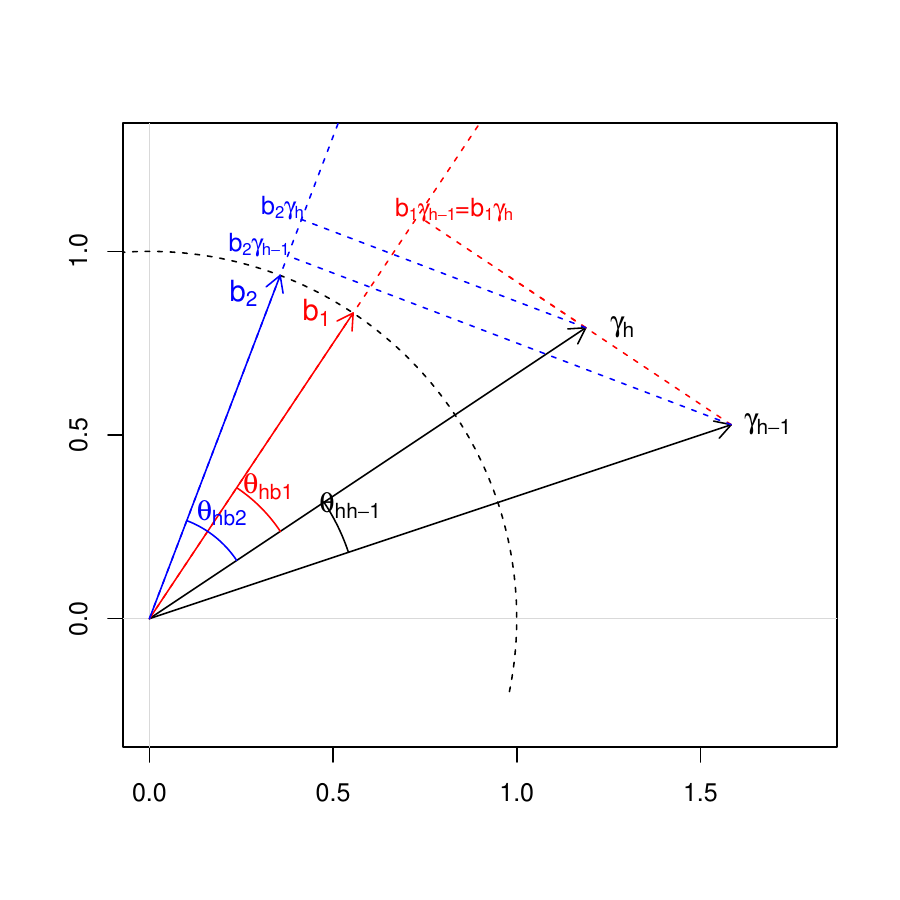}\caption{Geometry of the PCS predictor. Two solutions $\mathbf{b}_1$ (red) with $\beta_h=\mathbf{b}_1'\boldsymbol{\gamma}_{h-1}-\mathbf{b}_1'\boldsymbol{\gamma}_{h}=0$ and $\mathbf{b}_2$ (blue) with $\beta_h=\mathbf{b}_2'\boldsymbol{\gamma}_{h-1}-\mathbf{b}_2'\boldsymbol{\gamma}_{h}<0$, are shown in the plane spanned by $(\boldsymbol{\gamma}_{h-1},\boldsymbol{\gamma}_h)$. The positive angles $\theta_{hb1}$ and $\theta_{hb2}$ denote the phase excesses of the respective PCS predictors relative to the MSE predictor $\boldsymbol{\gamma}_h$. \label{pcs_geometry}}\end{center}\end{figure}

\subsection{The Construction of Leading Indicators and Benchmark Customization}

The (MSE-)PCS objective in \eqref{pcs_mse} can be adapted to tasks such as constructing macroeconomic leading indicators or customizing benchmarks. Since these settings share the same structure, we focus on the former.\\

Let $\boldsymbol{\gamma}$ be the causal, length-$L$ MA inversion of the target (coincident) indicator, e.g., the Wold decomposition of a stationary macro series (such as differenced GDP), or the causal right-half of a (possibly acausal) trend/cycle filter applied to that Wold decomposition (for brevity we do not discuss multivariate designs). We then solve
\begin{eqnarray}
&&\min_{\mathbf{b}} (\boldsymbol{\gamma}-\mathbf{b})'(\boldsymbol{\gamma}-\mathbf{b})\label{pcs_leading}\\
&&\left((\mathbf{F}^{h-1}-\mathbf{F}^h)\boldsymbol{\gamma}\right)'\mathbf{b}=\beta_h.\nonumber
\end{eqnarray}
Unlike the original formulation \eqref{pcs_mse},  criterion \eqref{pcs_leading} does not aim to approximate the $h$-step MSE predictor. Instead, it targets the impulse response $\boldsymbol{\gamma}$ of the indicator while enforcing a constraint that moves the CCF peak $h$ periods to the right, inducing leading behavior. The solution lies in the span of $\boldsymbol{\gamma}$ and $(\mathbf{F}^{h-1}-\mathbf{F}^h)\boldsymbol{\gamma}$, and therefore generally differs from DFP and the original PCS designs.\\

Benchmark customization means tracking a chosen target $\boldsymbol{\gamma}$ as closely as possible while adjusting its timeliness to satisfy the PCS constraint. This preserves a preferred design while allowing favorable real-time fine-tuning; see Wildi (2024–2026) for smoothness customizations.\\
When $\boldsymbol{\gamma}:=\boldsymbol{\gamma}_h$, PCS customizes the MSE predictor, but $\boldsymbol{\gamma}$ can be any target filter, including standard business-cycle tools such as Hodrick–Prescott (see Section \ref{RTBCA}), Christiano–Fitzgerald, Beveridge–Nelson, or Hamilton filters.

\subsection{Accuracy-Timeliness Dilemma}\label{ac_ti_di}

For the unit-length DFP predictor in \eqref{dp}, lead (advancement) appears as a positive phase excess, $\theta_{0b}>\theta_{0h}$ in Fig. \ref{dfp_geometry} (a formal connection between lead and phase excess is established in the next Section \ref{time_shift}).  In this setting, the decoupling parameter and phase are linked by $\alpha_0=\cos(\theta_{0b})$. Expressing the decoupling constraint in terms of the phase renders the accuracy–timeliness trade-off explicit: accuracy is encoded in the objective, while timeliness is governed by the phase-based constraint.\\
Geometrically, as the phase excess $\theta_{0b}-\theta_{0h}$ increases (with $\theta_{0b}>\theta_{0h}$ throughout)) in Fig.\ref{dfp_geometry}, the projection of $\mathbf{b}$ onto $\boldsymbol{\gamma}_h$ declines. Under the unit-norm constraint $\|\mathbf{b}\|=1$,  
\[
\mathbf{b}'\boldsymbol{\gamma}_h = \cos(\theta_{0b}-\theta_{0h})\|\boldsymbol{\gamma}_h\|,
\]
so a larger phase excess reduces the objective value $\mathbf{b}'\boldsymbol{\gamma}_h$ (and vice versa). 
An analogous argument applies to the PCS predictor, where the objective $\mathbf{b}'\boldsymbol{\gamma}_h=\cos(\theta_{hb})\|\boldsymbol{\gamma}_h\|$ decreases as the phase excess $\theta_{hb}$  increases (a formal link between $\theta_{hb}$ and the hyperparameter $\beta_h$ is established in Appendix \ref{phe_bet}, Proposition \ref{l_t_b}). This captures the inherent dilemma between   timeliness (phase excess or lead) and accuracy (target correlation or MSE).

\subsection{Dual Interpretation and Efficient (Accuracy-Timeliness) Frontier}


Consider the dual DFP optimization criterion 
\begin{eqnarray}
&&\min_{\mathbf{b}}\boldsymbol{\gamma}_{0}'\mathbf{b}\label{dp_dual}\\
\textrm{s.t.}&& \boldsymbol{\gamma}_h'\mathbf{b}/\|\boldsymbol{\gamma}_h\|=\alpha_h\nonumber\\
&&\mathbf{b}'\mathbf{b}=1\nonumber.
\end{eqnarray}
It is obtained from the primal DFP problem by interchanging the objective and the correlation constraint and switching to minimization. While the algebra is analogous, the geometry changes: the feasible cone determined by the constraints is centered on $\boldsymbol{\gamma}_h$ (not $\boldsymbol{\gamma}_0$),  i.e., the cone’s boundary rays are symmetric about $\boldsymbol{\gamma}_h$.\\

Let $\theta_{primal}=\theta_{0b}$ be the angle between $\mathbf{b}$ and $\boldsymbol{\gamma}_0$ in the primal geometry (cf. Fig.\ref{dfp_geometry}). In the dual problem, the constraint fixes  $\theta_{dual}=\theta_{hb}$, the angle between $\mathbf{b}$ and $\boldsymbol{\gamma}_h$. Since $\|\mathbf{b}\|=1$, the correlation constraint gives $\alpha_h=\cos(\theta_{dual})$. Geometrically,  $\theta_{0b}=\theta_{0h}+\theta_{hb}$, where $\theta_{0h}$ is the angle between $\boldsymbol{\gamma}_0$ and $\boldsymbol{\gamma}_h$. Therefore, 
\[
\theta_{primal}=\theta_{dual}+\theta_{0h}=\arccos(\alpha_h)+\theta_{0h}~,~ \alpha_0=\cos(\theta_{primal})=\cos(\arccos(\alpha_h)+\theta_{0h}).
\]
Thus $\alpha_0$ (primal) and $\alpha_h$ (dual) are linked via this identity. The dual solutions corresponding to $\pm\theta_{dual}$ yield objective values
\[
\boldsymbol{\gamma}_0'\mathbf{b} = \cos(\pm\theta_{dual}+\theta_{0h})\|\boldsymbol{\gamma}_0\|
\]
in \eqref{dp_dual}. 
Selecting the minimizing branch in \eqref{dp_dual} yields $+\theta_{dual}$, matching the original DFP solution, with $\theta_{primal}=\arccos(\alpha_h)+\theta_{0h}$. Hence, when $\alpha_0:=\cos(\arccos(\alpha_h)+\theta_{0h})$ (equivalently $\alpha_h:=\cos(\arccos(\alpha_0)-\theta_{0h})$), the primal and dual formulations coincide, giving a dual interpretation of DFP: it is the `fastest' predictor—i.e., the one with maximal positive phase excess $\theta_{dual}$—subject to a prescribed tracking accuracy
$\alpha_h$, measured by the correlation $\boldsymbol{\gamma}_h'\mathbf{b}/\|\boldsymbol{\gamma}_h\|$  with the MSE predictor $\boldsymbol{\gamma}_h$.\\

The equivalence of primal and dual formulations demonstrates that the DFP criterion creates an efficient frontier between accuracy (target correlation/MSE) and timeliness (phase excess/lead). Unlike the MSE predictor, which occupies only one spot on this curve, the DFP criterion spans the full frontier. The PCS approach follows a similar rationale, though it defines timeliness via a different metric (CCF shift) instead of phase excess.

\section{Connecting Hyperparemeters and Phase-Excess with Lead}\label{time_shift}

For the PCS predictor, the concept of a lead is directly associated with a peak shift in its cross-correlation function (CCF). Specifically, a rightward displacement of the CCF peak constitutes a quantifiable lead over the MSE predictor (De Jong and Nijman, 1997). Conversely, the phase excess $\theta_{hb}>0$ exhibited by the DFP predictor does not lend itself to an immediate interpretation as a temporal lead. To establish a functional relationship between $\theta_{hb}$ (or the hyperparameter $\alpha_0$) and the relative lead of the DFP predictor over the MSE benchmark, we must first define the concept of a lead. Notably, this definition is generalizable and can be applied equally to the PCS predictor.

\subsection{Time-Shift at Frequency Zero}

Define the (complex) transfer function of the length-$L$ filter $\boldsymbol{\gamma}=(\gamma_{0},...,\gamma_{L-1})'$ by
\[
\Gamma(\omega)=\sum_{k=0}^{L-1}\gamma_{k}\exp(-ik\omega),~ \omega\in [0,\pi]\footnote{Negative frequencies $\omega\in [-\pi,0[$ can be ignored due to symmetry.}.
\]
Write its polar form as
\[
\Gamma(\omega)=A(\omega)\exp(i\Phi(\omega)),
\]
where $A(\omega)=|\Gamma(\omega)|\geq 0$  is the amplitude response and $\Phi(\omega)=\textrm{arg}(\Gamma(\omega))$ is the phase response (mod$(2\pi)$).  For $\omega\neq 0$, define the associated time-shift 
\begin{eqnarray}\label{ts_def}
\tau(\omega):=\Phi(\omega)/\omega.
\end{eqnarray}
Let $x_t=\exp(it\omega)$. The filtered output is  
\[
y_t=\sum_{k=0}^{L-1}\gamma_kx_{t-k}=\sum_{k=0}^{L-1}\gamma_k\exp(i(t-k)\omega)=\Gamma(\omega)x_t=A(\omega)\exp(i(t+\tau(\omega))\omega)=A(\omega)x_{t+\tau(\omega)}.
\]
Hence, when a sinusoid $x_t=\sin(t\omega)$  passes through the filter,  the output is scaled by $A(\omega)$ and time-shifted by $\tau(\omega)$:  $\tau(\omega)>0$ indicates an advance (lead,left shift) and $\tau(\omega)<0$  indicates a retardation (lag,right shift). If $A(\omega)=0$, the sinusoid at frequency $\omega$ is completely suppressed, so the notion of a time-shift at that $\omega$ is not meaningful. \\
To extend the definition to $\omega=0$, differentiate $\Gamma(\omega)$ with respect to $\omega$:
\[
\dot{\Gamma}(\omega)=\dot{A}(\omega)\exp(i\Phi (\omega))+iA (\omega)\exp(i\Phi (\omega))\dot{\Phi} (\omega),
\]
where the dot denotes the first derivative with respect to $\omega$. Enforcing $\Gamma (0)>0$ rules out signal suppression or phase inversion, which yields $A(0)=\Gamma(0)>0$ and $\Phi(0)=0$. Since $A$ is even, $\dot{ A}(0)=0$, hence
\[
\dot{\Gamma}(0)=i \Gamma(0)\dot{\Phi}(0), \textrm{ ~so~} \dot{\Phi}(0)=-i\frac{\dot{\Gamma}(0)}{\Gamma(0)}.
\]
Consequently, the time-shift at $\omega=0$ is obtained from  
\begin{eqnarray}\label{ts_0}
\tau(0):=\lim_{\omega\to 0}\tau(\omega)=\lim_{\omega\to 0}\frac{\Phi(\omega)}{\omega}=\lim_{\omega\to 0}\frac{\Phi(\omega)-0}{\omega-0}=\dot{\Phi}(0)=-i\frac{\dot{\Gamma}(0)}{\Gamma(0)}=-\frac{\sum_{k=1}^{L-1}k\gamma_{k}}{\sum_{k=0}^{L-1}\gamma_{k}}.
\end{eqnarray}
If the filter is applied to a zero-frequency linear trend $x_t=t$, then
\[
y_t=\sum_{k=0}^{L-1}\gamma_{k}x_{t-k}=\sum_{k=0}^{L-1}\gamma_{k}(t-k)=t\sum_{k=0}^{L-1}\gamma_{k}-\sum_{k=0}^{L-1}k\gamma_{k}=(t+\tau(0))\sum_{k=0}^{L-1}\gamma_{k}=A(0)x_{t+\tau(0)},
\]
which confirms Equation \eqref{ts_0}. Since $\tau(\omega)$ varies with $\omega$, we fix a reference frequency $\omega_0$, taking $\omega_0=0$ to emphasize low-frequency (trend) components relevant for forecasting/nowcasting (any $\omega_0$ is possible via \eqref{ts_def}). We assume a `standard' scenario in which the MSE predictor leads the nowcast, $\tau_h(0) > \tau_0(0)$ (the non-standard case $\tau_h(0) \leq \tau_0(0)$ is addressed in Appendix \eqref{l_ns_s}). If, at this zero frequency, the DFP predictor also leads the MSE predictor, then $\tau_b(0) > \tau_h(0) > \tau_0(0)$; we later show how to tune $\lambda$ within $\mathbf{b}(\lambda)$ to hit a desired $\tau_b(0)$. This ordering aligns the frequency-domain triangle in the complex plane (with vertices $\mathbf{0}$, $\Gamma_h(\omega)$, $\Gamma_b(\omega)$) with the time-domain DFP triangle in Fig. \ref{dfp_geometry_2}, simplifying the exposition; see details below.  \\

Let $\Gamma_0(\omega), \Gamma_{h}(\omega),\Gamma_{b}(\omega)$ denote the transfer functions of the nowcast and the MSE and DFP predictors, with corresponding amplitude and phase functions $A_0,A_h,A_b$ and $\Phi_0,\Phi_h,\Phi_b$. Throughout this section, we assume that $\Gamma_{0}(0), \Gamma_{h}(0)$ and $\Gamma_{b}(0)$ are strictly positive at $\omega_0=0$, thereby ruling out signal suppression or `trend reversal'. In cases where this does not hold, an alternative reference frequency $\omega_0 > 0$ may be chosen, though we do not discuss that case here.



\subsection{Linking Hyperparameters to the Time-Shift}\label{linking}

We analyze the DFP-MSE predictor in \eqref{mse_dfp}, given by $\mathbf{b}=\boldsymbol{\gamma}_h+\lambda\boldsymbol{\gamma}_0$, and begin by linking $\lambda$ to the phase gap $\Phi_b(\omega)-\Phi_h(\omega)$ between the DFP and  MSE predictors. The same reasoning  applies to the PCS predictor and is therefore omitted. Unless stated otherwise, we assume the following baseline assumptions: $\boldsymbol{\gamma}_0\not\propto\boldsymbol{\gamma}_h$ (rank two),  $\theta_{0b}>\theta_{0h}>0$ (phase excess), $\Gamma_0(0)>0$,  $\Gamma_h(0)>0$,  $\Gamma_b(0)>0$ (no trend reversal) and $\tau_h(0)>\tau_0(0)$ (standard scenario).  

\begin{Proposition}\label{alpha_0_beta}
Let the baseline assumptions hold and define $\beta(\omega):=\Phi_b(\omega)-\Phi_h(\omega)$. Then, for $\omega$ in a neighborhood of zero,
\begin{eqnarray}\label{alpha0_shift}
\beta(\omega)=\arctan\left(\frac{\sin\left(\gamma(\omega)\right)}{a(\omega)/b(\omega)-\cos\left(\gamma(\omega)\right)}\right)\approx\frac{\gamma(\omega)}{a(0)/b(0)-1},
\end{eqnarray}
where $\gamma(\omega):=\Phi_h(\omega)-\Phi_0(\omega)$, $b=|\lambda|A_0(\omega)$, and $a=A_h(\omega)$. 
\end{Proposition}

\textbf{Proof}: From $\mathbf{b}=\boldsymbol{\gamma}_h+\lambda\boldsymbol{\gamma}_0$ it follows that 
\[
\Gamma_b(\omega)=\Gamma_h(\omega)+\lambda\Gamma_0(\omega). 
\]
Moreover, $\theta_{0b}>\theta_{0h}>0$ implies $\lambda< 0$ (cf. Fig.\ref{dfp_geometry_2}). Hence, the three transfer functions form a frequency-dependent DFP-triangle in the complex plane with side lengths $a(\omega)=A_h(\omega)$, $b(\omega)=|\lambda|A_0(\omega)$, and $c(\omega)=A_b(\omega)$\footnote{\label{foot_ao}The fixed time-domain vectors $\mathbf{b}$,$\boldsymbol{\gamma}_0$ and $\boldsymbol{\gamma}_h$ in the figure are replaced by their  transfer-function counterparts which now depend on $\omega$. Moreover, the abscissa and ordinate correspond to the real and imaginary axis, respectively.}. For $|\omega|$ sufficiently small (and $\lambda<0$), $\Gamma_h(\omega)$ lies between $\Gamma_0(\omega)$ and $\Gamma_b(\omega)$, i.e.,
\begin{eqnarray}\label{phase_ord}
\Phi_b(\omega)\geq\Phi_h(\omega)\geq\Phi_0(\omega).
\end{eqnarray}
The second inequality follows from 
\[
\Phi_h(\omega)\approx\dot{\Phi}_h(0)\omega=\tau_h(0)\omega>\tau_0(0)\omega=\dot{\Phi}_0(h)\omega\approx\Phi_0(\omega)
\]
where inequality and Taylor expansions are justified by the baseline assumptions $\Gamma_k(0)>0$ (so $\Phi_k(0)=0$) and $\tau_h(0)>\tau_0(0)$. The first inequality in \eqref{phase_ord}, $\Phi_b(\omega)\geq\Phi_h(\omega)$, is implied by $\lambda<0$ (reflecting phase excess $\theta_{0b}>\theta_{0h}$).  When $\lambda<0$, the ordering in \eqref{phase_ord} continues to hold as long as $\Phi_h(\omega)\leq\Phi_0(\omega)+\pi$.\footnote{Geometrically, for  $\lambda\in]-\infty,0]$, the argument $\Phi_b$ of $\Gamma_b=\Gamma_h+\lambda\Gamma_0$ lies in $[\Phi_h,\Phi_0+\pi]$, with the upper bound attained as $\lambda\to-\infty$. If $\Phi_h>\Phi_0+\pi$, then $\Phi_b<\Phi_h$ which affects the ordering in \eqref{phase_ord}.} This requirement holds for $\omega$ sufficiently close to zero: the phases are well-defined at $\omega=0$ (all transfer functions are non-vanishing), continuous in $\omega$, and obey $\Phi(0)=0$ because the transfer functions are strictly positive at zero frequency (no trend reversal). \\
Assume first the \emph{strict} ordering $\Phi_b(\omega)>\Phi_h(\omega)>\Phi_0(\omega)$ in \eqref{phase_ord}  which yields a well-defined, nondegenerate triangle. The  angles opposite the sides $a,b,c$ are $\alpha(\omega)$, $\beta(\omega)=\Phi_b(\omega)-\Phi_h(\omega)$ (i.e., $\theta_{0b}-\theta_{0h}$ in Figure \ref{dfp_geometry_2}), and $\gamma(\omega)=\Phi_h(\omega)-\Phi_0(\omega)$ (i.e., $\theta_{0h}$). 
Given the side lengths $a(\omega),b(\omega)$ and their included angle $\gamma(\omega)$, we have 
\begin{equation}\label{atan_f}
\beta(\omega)=\arctan\left(\frac{\sin(\gamma(\omega))b(\omega)}{a(\omega)-\cos(\gamma(\omega))b(\omega)}\right)=\arctan\left(\frac{\sin(\gamma(\omega))}{a(\omega)/b(\omega)-\cos(\gamma(\omega))}\right)\approx \frac{\gamma(\omega)}{a(0)/b(0)-1},
\end{equation}
where the last step holds for sufficiently small $|\omega|$, using first-order Taylor expansions of the trigonometric terms. This approximation is well-defined because
\begin{equation}\label{well_def}
a(0)/b(0)-1=\Gamma_h(0)/(|\lambda|\Gamma_0(0))-1\neq 0;
\end{equation}
otherwise $|\lambda|=\Gamma_h(0)/\Gamma_0(0)$ would (with $\lambda<0$) force $\Gamma_b(0)=0$. \\
Next, fix $\omega$ sufficiently close to zero and assume $\gamma(\omega)=\Phi_h(\omega)-\Phi_0(\omega)=0$. In this case the triangle degenerates into a line segment, so $\beta(\omega)\in\{0,\pi\}$ depending on $\lambda$. In particular, if $\lambda<0$ (phase excess) and $|\lambda|<\|\boldsymbol{\gamma}_h\|/\|\boldsymbol{\gamma}_0\|$, then the side lengths satisfy $b(\omega)<a(\omega)$, which implies $\beta_h=0$. If instead $|\lambda|\geq\|\boldsymbol{\gamma}_h\|/\|\boldsymbol{\gamma}_0\|$, then $\beta(\omega)$ is either undefined or equal to $\pi$.\\
Since phase reversal by the DFP predictor is excluded at $\omega=0$ (all transfer functions are strictly positive), the same `no phase reversal' condition, $\arg(\Gamma_b(\omega))-\arg(\Gamma_h(\omega))\neq \pm\pi$, must hold at the fixed $\omega$, due to continuity (provided $\omega$ is sufficiently close to zero). Therefore, the second possibility $\beta(\omega)=\pi$ is ruled-out and  $\beta(\omega)=0$ must hold, which verifies Equation \eqref{atan_f} in the first degenerate case. A similar argument applies when $\beta(\omega)=\Phi_b(\omega)-\Phi_h(\omega)=0$ (second degenerate case: the triangle collapses again to a line segment): for $\lambda<0$ this implies that the side-lengths satisfy $c(\omega)<a(\omega)$, implying $\gamma(\omega)=\Phi_h(\omega)-\Phi_0(\omega)=0$, so the arctan term in \eqref{atan_f} is zero, completing the proof.
\hfill\qed\\

If $L$ is sufficiently large, then $\boldsymbol{\gamma}_h\approx\mathbf{F}^h\boldsymbol{\gamma}_0$, so  $\Phi_h(\omega),A_h(\omega)$ are fully determined by $\boldsymbol{\gamma}_0$ (up to negligible deviations). Consequently, $\beta(\omega)$ in \eqref{alpha0_shift} depends only on $\lambda$ and $\boldsymbol{\gamma}_0$.\\

\textbf{Remark}: A key technical difficulty is that the frequency-domain triangle with vertices $\mathbf{0}$, $\mathbf{\Gamma}_h$ and $\mathbf{\Gamma}_b$ (in the complex plane) depends on $\omega$ and therefore deforms with $\omega$, unlike the fixed time-domain triangle in Fig.~\ref{dfp_geometry_2}. We thus restrict $\omega>0$ in a sufficiently small neighborhood of zero so the deformation is limited and the ordering in \eqref{phase_ord} remains unchanged. This ordering is anchored by the baseline assumptions, which fix the limiting configuration as $\omega\to 0$ approaches the reference frequency; at $\omega=0$  the frequency-domain triangle degenerates to a line segment.\\

The proposition's (final) baseline assumption $\tau_h(0)>\tau_0(0)$ is not essential for deriving $\Phi_b(\omega)-\Phi_h(\omega)$, but it streamlines the analysis and the derivation. Appendix \eqref{l_ns_s} treats the nonstandard case $\tau_h(0)<\tau_0(0)$. 
Obtaining larger leads would then typically require either flipping signs (turning a maximization into a minimization) or swapping the nowcast and MSE predictor in the objective and constraint, which runs counter to common sense.\\

Next, we connect $\lambda$ to the relative shift $\tau:=\tau_b(0)-\tau_h(0)$ of the DFP-MSE predictor  with respect to the MSE benchmark at frequency zero  (the same approach extends to the unit-length DFP, the PCS predictors, and nonzero frequencies). A relative lead at frequency zero is produced by enforcing $\tau>0$; we now show how to choose $\lambda=\lambda(\tau)$ as a function of $\tau$.

\begin{Theorem}\label{alpha0_rho0_cor}  
Let the baseline assumptions hold and let $\tau_b(0)-\tau_h(0)=\tau>0$ denote a target lead of the DFP relative to the MSE predictor at frequency zero, with $\tau\neq \tau_0(0)-\tau_h(0)$. Then, $\lambda$ is determined by $\tau$ as
\begin{eqnarray}\label{lambda0_tau}
\lambda=-\frac{\tau \Gamma_h(0)}{(\tau+\tau_h(0)-\tau_0(0))\Gamma_0(0)}.
\end{eqnarray}
\end{Theorem}

\textbf{Proof}: Observe that 
\begin{eqnarray*}
\tau&=&\tau_{b}(0)-\tau_{h}(0)=\lim_{\omega\to 0}\left(\tau_{b}(\omega)-\tau_{h}(\omega)\right)=\lim_{\omega\to 0}\left(\Phi_{b}(\omega)/\omega-\Phi_{h}(\omega)/\omega\right)\\
&=&\lim_{\omega\to 0}\frac{\beta(\omega)}{\omega}=\lim_{\omega\to 0}\frac{\gamma(\omega)/\omega}{A_h(0)/(|\lambda|A_0(0))-1}=\frac{\tau_h(0)-\tau_0(0)}{\Gamma_h(0)/(|\lambda|\Gamma_0(0))-1},
\end{eqnarray*}
where $\beta(\omega),\gamma(\omega)$ are as in Proposition \ref{alpha_0_beta}.  In the present setting, with $\omega\to 0$, the Taylor approximation on the right hand side of \eqref{alpha0_shift} is exact. The quotient is well-defined when $\Gamma_b(0)>0$; see \eqref{well_def}. Solving for $\lambda$ gives
\begin{eqnarray*}
\lambda=-\frac{\tau \Gamma_h(0)}{(\tau+\tau_h(0)-\tau_0(0))\Gamma_0(0)},
\end{eqnarray*}
This is well defined because, by assumption,  $\tau\neq \tau_0(0)-\tau_h(0)$. Moreover, under the baseline (standard-case) assumptions, $\lambda$ must be negative to induce a lead.\\
An interesting situation occurs when $\tau_0(0)-\tau_h(0)= 0$ (which would contradict the final baseline assumption).  
Then the formula for $\lambda$ simplifies to 
\[
\lambda= -\Gamma_h(0)/\Gamma_0(0),
\]
so it no longer depends on $\tau$. However, 
\begin{equation}\label{as_pred}
\Gamma_b(0)=\Gamma_h(0)+\lambda\Gamma_0(0)=0,
\end{equation}
which contradicts the assumption $\Gamma_b(0)>0$. We may therefore assume $\tau_0(0)-\tau_h(0)\neq 0$, completing the proof. \\
It is nevertheless instructive to  interpret this degenerate case geometrically. Since the derivative of phase differences 
\[
\dot{\Phi}_h(0)-\dot{\Phi}_0(0)=\tau_h(0)-\tau_0(0)= 0
\]
vanishes at $\omega=0$, the constant and linear terms in the Taylor expansion of $\Phi_h(\omega)-\Phi_0(\omega)$ cancel (the constant term is zero because the transfer functions are strictly positive), and  we obtain
\[
\Phi_h(\omega)-\Phi_0(\omega)=\textrm{O}(\omega^2).
\]
In the DFP-triangle with vertices $\mathbf{0},\Gamma_b(\omega),\Gamma_h(\omega)$ (cf. Fig.\ref{dfp_geometry_2} and footnote \ref{foot_ao}), the angle at $\Gamma_h(\omega)$ is
\[
\gamma(\omega)=\Phi_h(\omega)-\Phi_0(\omega)=O(\omega^2),
\]
i.e., the triangle becomes `super flat' as $\omega\to 0$ (in Fig.~\ref{dfp_geometry_2},  $\gamma(\omega)=\Phi_h(\omega)-\Phi_0(\omega)$ corresponds to $\gamma=\theta_{0h}$). 
If $\Gamma_b(0)> 0$ (as assumed), then $\beta(\omega)=\Phi_b(\omega)-\Phi_h(\omega)$ must be of the same order, i.e., $\beta(\omega)=O(\omega^2)$\footnote{Note that $\beta(\omega)=\pi+O(\omega^2)$ is excluded because $\Gamma_b(0)>0$ (no phase reversal towards $\omega=0$).}. This would imply
\[
\tau=\tau_b(0)-\tau_h(0)=\lim_{\omega\to 0}(\dot{\Phi}_b(\omega)-\dot{\Phi}_h(\omega))=0,
\] 
contradicting $\tau>0$. Therefore, $\Gamma_b(0)=0$, as predicted by \eqref{as_pred}. This analysis shows that one cannot have simultaneously $\Gamma_b(0)>0$ and $\tau>0$ when $\tau_h(0)-\tau_0(0)= 0$ (which is ruled-out by our assumptions).  \hfill\qed\\

\textbf{Remark}:  The restriction $\tau\neq \tau_0(0)-\tau_h(0)$ in the theorem guarantees that the DFP predictor does not offset or eliminate the lead  $\tau_h(0)-\tau_0(0)$ that the MSE benchmark has over the nowcast, which is a natural and intuitively desirable requirement.\\

For $\tau=0$, Eq.\eqref{lambda0_tau} gives $\lambda=0$ since $\tau_h(0)>\tau_0(0)$ (baseline assumption), and thus $\mathbf{b}=\boldsymbol{\gamma}_h$, as expected. As $\tau\to\infty$, $\lambda\to-\Gamma_h(0)/\Gamma_0(0)$, so 
\[
\Gamma_b(0)=\Gamma_h(0)+\lambda\Gamma_0(0)\to 0,
\]
i.e., the DFP predictor asymptotically removes linear trends. Although this determines the limit of  $\lambda=\lambda(\tau)$, as a function of $\tau$, one may choose $\lambda<-\Gamma_h(0)/\Gamma_0(0)$ if trend reversal is permitted; then the left-shift typically keeps increasing (see the next section), but $\tau_b(0)$ is no longer well-defined at zero frequency (Equation \eqref{ts_0} instead gives the time-shift of the sign-reversed filter).\\

We can also express $\alpha_0$ in the DFP constraint, as well as the phase excess $\theta_{hb}$ between the DFP predictor and the MSE predictor, in terms of  $\tau=\tau_b(0)-\tau_h(0)$.

\begin{Corollary}\label{sel_alpha}
Under the assumptions of the theorem
\begin{eqnarray}\label{alpha0_tau}
\alpha_0&=&\boldsymbol{\gamma}_0'\left(\boldsymbol{\gamma}_h+\lambda(\tau)\boldsymbol{\gamma}_0\right)\\
\theta_{hb}&=&\arccos\left(\frac{\boldsymbol{\gamma}_h'(\boldsymbol{\gamma}_h+\lambda(\tau)\boldsymbol{\gamma}_0}{\|\boldsymbol{\gamma}_h\|\|\boldsymbol{\gamma}_h+\lambda(\tau)\boldsymbol{\gamma}_0\|}\right)\approx 
\arccos\left(\frac{(\mathbf{F}^h\boldsymbol{\gamma}_0)'\left(\mathbf{F}^h+\lambda(\tau)\mathbf{I}\right)\boldsymbol{\gamma}_0}{\|\mathbf{F}^h\boldsymbol{\gamma}_0\|\|(\mathbf{F}^h+\lambda(\tau)\mathbf{I})\boldsymbol{\gamma}_0\|}\right)
,\label{thet_lam}
\end{eqnarray}
where $\lambda(\tau)$ is defined in \eqref{lambda0_tau}.
\end{Corollary}
The result follows from \eqref{lambda_mse} and the identity 
\[
\cos(\theta_{hb})=\frac{\boldsymbol{\gamma}_h'\mathbf{b}}{\|\boldsymbol{\gamma}_h\|\|\mathbf{b}\|}=\frac{\boldsymbol{\gamma}_h'(\boldsymbol{\gamma}_h+\lambda(\tau)\boldsymbol{\gamma}_0)}{\|\boldsymbol{\gamma}_h\|\|\boldsymbol{\gamma}_h+\lambda(\tau)\boldsymbol{\gamma}_0\|},
\]
using $\boldsymbol{\gamma}_h\approx \mathbf{F}^h\boldsymbol{\gamma}_0$ for $L$ sufficiently large.
\\

At $\tau=0$, $\lambda(\tau)=0$ and $\cos(\theta_{0b})=\frac{\boldsymbol{\gamma}_0'\boldsymbol{\gamma}_h}{\|\boldsymbol{\gamma}_0\|\|\boldsymbol{\gamma}_h\|}$, so $\theta_{0b}=\theta_{0h}$ and $\theta_{hb}=0$ (no phase excess because $\mathbf{b}=\boldsymbol{\gamma}_h$). Since $\lambda(\tau)$ is strictly decreasing for $\tau\geq 0$, $\theta_{hb}$ in \eqref{thet_lam} is strictly increasing.  Thus the dual interpretation in Section \ref{ac_ti_di} carries over  from phase excess $\theta_{hb}$ to lead $\tau$: DFP  maximizes $\tau$ subject to given tracking-accuracy. By primal–dual equivalence, DFP traces the efficient frontier between timeliness $\tau$ and accuracy (max-target correlation or min-MSE). \\

The relationship between $\alpha_0$ and the lead time $\tau$ makes it possible to interpret the hyperparameter and to choose appropriate values using a formal decision rule, as illustrated in the applications below.

\subsection{Limits on the Attainable Phase Excess and Lead under Strict Positivity}\label{limit}

To illustrate, assume $\boldsymbol{\gamma}_0$ and $\boldsymbol{\gamma}_h$ are nearly (but not exactly) aligned. As $\lambda<0$ becomes more negative, $\mathbf{b}=\boldsymbol{\gamma}_h+\lambda\boldsymbol{\gamma}_0$ rotates to point almost opposite $\boldsymbol{\gamma}_h$, so pushing $\lambda$ too far to increase lead can (almost) flip the predictor’s sign. Componentwise, such sign flips can sometimes increase lead (see Section \ref{AR3}), but the limit is extreme, so a criterion is needed to rule it out.\\

A basic requirement for a meaningful DFP predictor is the strict positivity condition (Section \ref{dfp}): the target correlation must be strictly positive, $\boldsymbol{\gamma}_h'\mathbf{b}(\lambda)>0$, equivalently $\theta_{hb}<\pi/2$. Via $\lambda(\tau)$ in Equation \eqref{lambda0_tau}  and Equation \eqref{thet_lam}, this imposes bounds on admissible $\lambda$ (and hence $\tau$). We require
\[
0< \cos(\theta_{hb})=\frac{\boldsymbol{\gamma}_h'(\boldsymbol{\gamma}_h+\lambda\boldsymbol{\gamma}_0)}{\|\boldsymbol{\gamma}_h\|\|\boldsymbol{\gamma}_h+\lambda\boldsymbol{\gamma}_0\|},
\]
which implies $\boldsymbol{\gamma}_h'(\boldsymbol{\gamma}_h+\lambda\boldsymbol{\gamma}_0)>0$. If $\boldsymbol{\gamma}_h'\boldsymbol{\gamma}_0>0$ then
\[
\lambda>-\frac{\boldsymbol{\gamma}_h'\boldsymbol{\gamma}_h}{\boldsymbol{\gamma}_h'\boldsymbol{\gamma}_0}=:\lambda_{\lim},
\]
(with a similar bound $\lambda<\left|\boldsymbol{\gamma}_h'\boldsymbol{\gamma}_h/\boldsymbol{\gamma}_h'\boldsymbol{\gamma}_0\right|$ when $\boldsymbol{\gamma}_h'\boldsymbol{\gamma}_0<0$). If $\lambda_{\lim}\Gamma_0(0)+\Gamma_h(0)>0$, substituting into \eqref{lambda0_tau} yields the upper bound:
\[
\tau\leq-\frac{(\tau_h(0)-\tau0(0))\lambda_{\lim }\Gamma_0(0)}{\lambda_{\lim}\Gamma_0(0)+\Gamma_h(0)}=\tau_{\lim}>0.
\]
If instead $\lambda_{\lim}\Gamma_0(0)+\Gamma_h(0)\leq 0$, the resulting $\mathbf{b}(\lambda_{\lim})$ eliminates or reverses the trend, making $\tau_{\lim}$ ill-defined; in that case, strict positivity imposes no restriction on $\tau$ and any finite $\tau$ is admissible.\\ 

We contend that any predictive lead (`left-shift') achieved by violating positivity is ill-conditioned. While dropping positivity can seem to extend the forecast horizon (e.g., in periodic settings), positivity ensures the predictor’s dynamics stay aligned with those of the process at horizon $h$. Thus $\theta_{hb}<\pi/2$—equivalently the bounds $\lambda_{\lim}$ or $\tau_{\lim}$—
places a hard cap on how far the DFP predictor  $\mathbf{b}(\lambda(\tau))$  can genuinely outperform the MSE benchmark $\boldsymbol{\gamma}_h$ while remaining positively correlated with the target $x_{t+h}$. This coupling to $x_{t+h}$ prevents  pushing the lead towards future points whose dynamics diverge  from the target horizon, avoiding, among others, a  sign-flip.\\

By primal–dual equivalence, for any fixed target correlation, no linear predictor can attain a larger lead than DFP. Therefore, any lead beyond $\tau_{\lim}$ must violate strict positivity, making $\tau_{\lim}$ a universal upper bound on lead among linear predictors—and DFP achieves it. The same applies to PCS, where lead manifests as a rightward shift of the CCF peak.\\

Even stronger bounds on $\lambda$ and $\tau$ can be established by enforcing the stricter positivity rule $\mathbf{b}(\lambda(\tau))'\boldsymbol{\gamma}_0 > 0$, i.e., $\theta_{0b} < \pi/2$, see Equation \eqref{str_pos} (details omitted).

\section{Empirical Examples}\label{examples}

We illustrate several aspects of the proposed designs. First, we apply the unit-length DFP criterion \eqref{dp} to the MA process introduced in Section \ref{onemulti}, treating $\epsilon_t$ as observed, and  compare partial and complete decoupling. Second, we implement the MSE-DFP criterion \eqref{dp2} for AR and ARMA processes, where the innovations $\epsilon_t$ are unobserved and are recovered by inversion. We also construct a predictor targeting a prespecified lead over the MSE benchmark (Theorem \ref{alpha0_rho0_cor}) and evaluate finite-sample performance via simulation. Finally, we apply the leading indicator PCS \eqref{pcs_leading} to business-cycle analysis using quarterly GDP and the Hodrick–Prescott filter (Hodrick and Prescott, 1997).

\subsection{Illustration of the DFP-Criterion in the MA(9)-Case}\label{la}

We apply the unit-length DFP criterion \eqref{dp} to the MA(9) example with forecast horizon $h = 5$, treating $\epsilon_t$ as observed (see Section \ref{onemulti}). Figure \ref{cor} shows that the CCF $\rho(\hat{x}_{5t}^{MSE},x_{t+\delta})$ of the MSE predictor attains a peak of 0.86 at $\delta=0$, indicating a strong tie with the present $x_t$—i.e., a predominantly coincident predictor $\hat{x}_{5t}^{MSE}$. To mitigate this, we consider two DFP designs that attenuate the CCF at $\delta=0$: (i) partial decoupling, implemented by setting the CCF to $0.43$ (i.e., $\alpha_0 = 0.43$), and (ii) complete decoupling, imposing a vanishing CCF via $\alpha_0=0$. Solving the quadratic in \eqref{quad} and selecting the root that maximizes the objective yields the predictors
\begin{eqnarray*}
\mathbf{{b}}^{\textrm{partial}}&=&1.62\boldsymbol{\gamma}_h-0.51\boldsymbol{\gamma}_0\\
\mathbf{{b}}^{\textrm{complete}}&=&1.8\boldsymbol{\gamma}_h-0.79\boldsymbol{\gamma}_0.
\end{eqnarray*}

Figure \ref{cor_cor_LA_mod} compares predictors (top-left), CCF (top-right), and forecasts (bottom). Unlike the MSE design, the DFP predictor can place nonzero weights up to lag $k=L$ because $\mathbf{b}$ depends on $\boldsymbol{\gamma}_0$ (a similar result is reported in Wildi (2024) and (2026a), when prioritizing smoothness rather than lead). 
The DFP criterion enforces decoupling at $\delta=0$ (top-right) while maximizing forecast performance at the target horizon $h = 5$.  Ideally, any degradation at $\delta=0$ does not propagate to $\delta=h$; the criterion is meant to minimize such spillover at the target horizon. 
The apparent CCF peak shift to $\delta=h$ under full decoupling (top right, red) is incidental in this example. By contrast, PCS (which also shifts the peak to $\delta=h$) differs because its CCF does not vanish at $\delta=0$ (not shown). As decoupling strengthens, the DFP forecasts (bottom) shift progressively left, indicating relative anticipation. Notably, this lead is achieved by placing more weight on higher-lag observations, which is counterintuitive. \\

This example highlights the core trade-off underlying the DFP criterion: timeliness (advancement via decoupling) versus accuracy (MSE performance at $\delta=h$).  Table \ref{mse_la_decoupling} quantifies this trade-off through the CCF and the relative lead of the DFP design compared with the benchmark. Leads are measured as the shift at which the sample correlation between the MSE predictor and the (shifted) DFP predictor attains its maximum (cf. De Jong and Nijman, 1997).
\begin{table}[ht]
\centering
\begin{tabular}{rrrr}
  \hline
 & MSE & DFP weak decoupling & DFP complete decoupling \\ 
  \hline
CCF at lag=0 & 0.86 & 0.43 & 0.00 \\ 
  CCF at h=5 & 0.51 & 0.42 & 0.26 \\ 
  Relative lead over MSE & 0.00 & 2.00 & 4.00 \\ 
   \hline
\end{tabular}
\caption{CCFs of the predictors at $\delta=0$ and $\delta=h=5$ for the MSE benchmark, DFP with weak decoupling, and DFP with complete decoupling. Forecast loss at the target horizon (second row) is minimized subject to the decoupling constraint at the present (first row).  Empirical leads of the DFP predictors (third row) quantify the left shift (earlier timing) relative to the MSE benchmark, defined as the lag at which the sample correlation between the DFP and MSE predictors is maximized. The reduction in CCF at the forecast horizon alongside the increase in lead highlights the accuracy–timeliness trade-off.} 
\label{mse_la_decoupling}
\end{table}

\begin{figure}[H]\begin{center}\includegraphics[height=4in, width=6in]{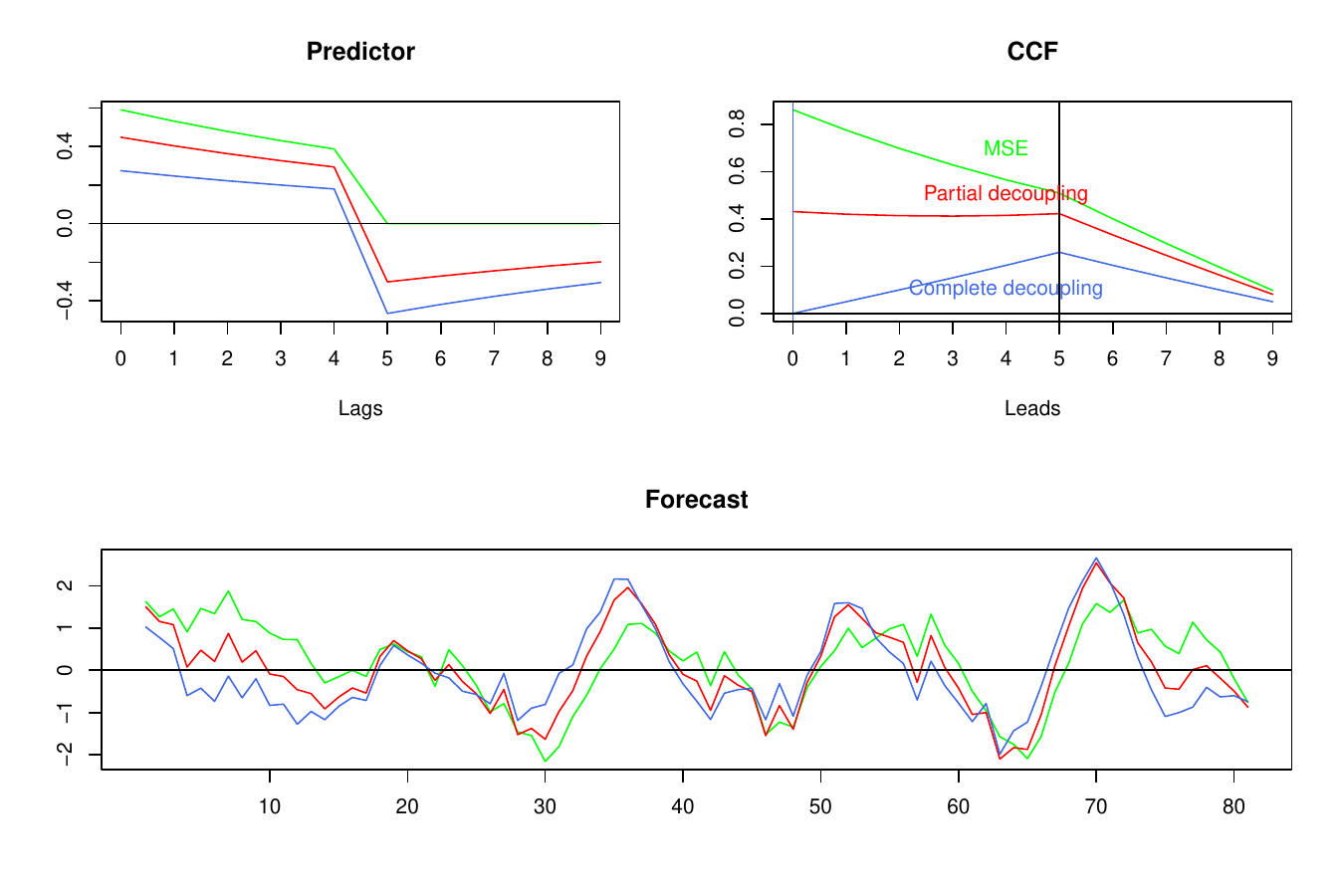}\caption{MSE (green) and DFP predictors under weak decoupling (red) and complete decoupling (blue). Top-left: forecast weights. Top-right: CCF across leads (the CCF under complete decoupling is zero at $\delta= 0$). Bottom: standardized forecasts. Increasing decoupling produces a leftward shift (advancement) of the DFP forecasts relative to the MSE benchmark.\label{cor_cor_LA_mod}}\end{center}\end{figure}

\subsection{Application of the Mean-Squared DFP to AR- and ARMA-Processes}\label{AR3}

We apply the mean-squared-error variant \eqref{dp2} of the DFP predictor to the stationary AR(3) and ARMA(3,2) processes
\begin{eqnarray*}
x_{t1}&=&1.3x_{t-1,1}-0.46x_{t-2,1}+0.048x_{t-3,1}+\epsilon_{t1}\\
x_{t2}&=&0.4x_{t-1,2}+0.3x_{t-2,2}+0.2x_{t-3,2}+\epsilon_{t2}+0.5\epsilon_{t-1,2}+0.4\epsilon_{t-2,2}.
\end{eqnarray*}
For the first process, the characteristic polynomial admits three positive real roots. In the ARMA(3,2) case, it admits one positive real root and a complex-conjugate pair whose modulus is substantially smaller than that of the real root. Consequently, both processes are dominated by the positive real root, yielding slowly and monotonically decaying autocorrelation functions (not shown) and predominantly non-oscillatory dynamics (an application to cyclical behavior is presented in the final business-cycle example).\\

In contrast to the preceding example, the innovation sequence $\epsilon_t$ is not observed and must be recovered by inversion to obtain the Wold decompositions
\[
x_{ti}=\sum_{k=0}^{\infty}\gamma_{ki}\epsilon_{t-k,i}, i=1,2.
\]
Because the dynamics are non-oscillatory, the MSE predictor tends to be strongly contemporaneously coupled, with a comparatively large CCF at $\delta=0$, even for fairly large forecast horizons. To attenuate this coupling, we employ the MSE DFP criterion \eqref{dp2} and vary the decoupling parameter $\alpha_0$ to construct $h=5$-step-ahead MSE and DFP predictors (see Fig. \ref{ar3_canc}). Under complete decoupling ($\alpha_0=0$), the solutions are
\begin{eqnarray*}
\mathbf{{b}}^{\textrm{AR(3)}}&=&\boldsymbol{\gamma}_h^{AR(3)}-0.41\boldsymbol{\gamma}_0^{AR(3)}\\
\mathbf{{b}}^{\textrm{ARMA(3,2)}}&=&\boldsymbol{\gamma}_h^{ARMA(3,2)}-0.76\boldsymbol{\gamma}_0^{ARMA(3,2)}\\
\end{eqnarray*}
As noted, a limitation of the MSE-based DFP variant \eqref{dp2} is that, in general, $\alpha_0$ lacks a correlation interpretation\footnote{That said, the optimization principle is simpler, yields a unique solution, and readily extends to nonstationary processes (not shown).} (except when $\alpha_0=0$). To address this, we can instead choose $\alpha_0$ based on a pre-specified time shift (advance) relative to the MSE predictor, as discussed later. Meanwhile,   Table \ref{ar3_decoupling} reports the effective CCF values at $\delta=0$ and $\delta=5$ implied by the selected $\alpha_0$ for both processes. The optimization pursues a demanding objective: minimize performance loss at $\delta=h$ while enforcing a controlled decline at $\delta=0$. This deliberate trade-off sets our method apart from classical forecasting approaches, which generally refrain from minimizing performance at zero lag—or at any lag (dual interpretation). \\

\begin{table}[ht]
\centering
\begin{tabular}{rrrrr}
  \hline
 &  $\textrm{Process 1: CCF }\delta=0$ & $\delta=5$ & $\textrm{Process 2: } \delta=0$ & $\delta=5$ \\ 
  \hline
$\alpha_0=0.9$ & 0.72 & 0.38 & 0.90 & 0.70 \\ 
  $\alpha_0=0.45$ & 0.46 & 0.29 & 0.72 & 0.57 \\ 
  $\alpha_0=0.22$ & 0.24 & 0.21 & 0.45 & 0.38 \\ 
  $\alpha_0=0.1$ & 0.11 & 0.16 & 0.22 & 0.21 \\ 
  $\alpha_0=0$ & 0.00 & 0.12 & 0.00 & 0.04 \\ 
   \hline
\end{tabular}
\caption{CCFs of the DFP predictors for the first process (columns 1–2) and the second process (columns 3–4) evaluated at $\delta=0$ (columns 1 and 3) and at $\delta=h=5$ (columns 2 and 4), shown for multiple values of the decoupling parameter $\alpha_0$. Note that $\alpha_0$ generally differs from the CCF at $\delta=0$ except in the case of complete decoupling (columns 1 and 3, last row).} 
\label{ar3_decoupling}
\end{table}

The contrast between the MSE (green) and DFP predictors for both processes in Fig.\ref{ar3_canc} highlights distinctive properties of the DFP design: whereas the MSE predictors are smooth, strictly positive, and monotonically decaying—reflecting the influence of the dominant root of the characteristic polynomial—the DFP predictors become progressively less smooth, partially negative, and non-monotonic as decoupling increases. Specifically, as $\alpha_0$ decreases, $\lambda$ in \eqref{lambda_mse} becomes more negative, intensifying the cancellation between $\boldsymbol{\gamma}_0$ and $\boldsymbol{\gamma}_h$ in the DFP expression $\mathbf{{b}}=\boldsymbol{\gamma}_h+\lambda\boldsymbol{\gamma}_0$. The consequent down-weighting of the present relative to the future in $\mathbf{{b}}$ exposes dynamics otherwise obscured by the dominant root. Beyond the increasingly irregular patterns observed for the ARMA(3,2) case (right top panel), a damped cycle emerges in the AR(3) case (left top panel) despite all process roots being strictly positive—a configuration that typically implies noncyclical behavior (cyclical dynamics are examined in the final example of the next section).\\

We now select the hyperparameter $\alpha_0$ using a pre-specified time-shift (see Corollary \ref{sel_alpha}) and compare the resulting procedure with the fully decoupled DFP; both are assessed against the MSE predictor. For illustration, we use the (above) AR(3) process
\[
x_t=1.3x_{t-1}-0.46x_{t-2}+0.048x_{t-3}+\epsilon_t.\]
We construct an $h=3$-step-ahead predictor and impose a relative time shift (advance) of $\tau=\tau_{b}(0)-\tau_h(0)=2$ with respect to the MSE predictor $\boldsymbol{\gamma}_{3}$ to obtain the new $\tau$-shifted DFP. While the particular forecast horizons used in this example are somewhat arbitrary, the conclusions carry over to any horizon. In particular, the resulting DFP predictors (based on $h=3$) largely preserve the left shift—even relative to the MSE predictor $\boldsymbol{\gamma}_{\tilde{h}}$—for arbitrarily long horizons $\tilde{h}$. See Table \ref{ar3_shift_decoupling} for an illustration with $\tilde{h}=20$.\\

Figure \ref{ar3_shift_full} (top left) displays the finite-length MA inversion $\boldsymbol{\gamma}_0$, i.e., the truncated Wold decomposition of the AR(3), (black) alongside with the MSE predictor  $\boldsymbol{\gamma}_{3}$ (green), the $\tau$-shifted DFP (blue), and the fully decoupled DFP (red). For ease of visual inspection, all designs are calibrated to unit-length; hence the Wold decomposition (black) does not start with a one.  The associated CCFs and sample realizations are presented in the top-right and bottom panels, respectively. Table \ref{ar3_shift_decoupling} lists the time-shifts $\tau(\omega)$ and transfer functions $\Gamma(\omega)$ at $\omega=0$, along with the hyperparameters $\lambda$ and $\alpha_0$ ($\Gamma(0)$ is based on predictors scaled to unit-length). For comparison, we also evaluate an additional 
20-step-ahead MSE predictor to confirm that increasing the forecast horizon only marginally changes the time shift. Hence, the `hyperparameter' $h$ (the forecast horizon) cannot be used to address a left shift (lead) of the MSE predictor.   The fully decoupled DFP is obtained by setting $\alpha_0=0$, with $\lambda$ computed from \eqref{lambda_mse}. The $\tau$-shifted DFP uses $\tau=2$, with $\lambda$ and $\alpha_0$ given by \eqref{lambda0_tau} and  \eqref{alpha0_tau}, respectively. \\

Because the fully decoupled DFP undergoes phase reversal at zero frequency ($\Gamma(0)<0$), its time shift is not well-defined and is omitted from the table.\footnote{In particular, the formula for $\tau(0)$ returns the shift of the sign-reversed predictor.} While the $\tau$-shifted DFP advances trends by $\tau=2$ relative to the MSE predictor, the fully decoupled version does not preserve the direction of a linear trend.\\

This inversion also reverses any nonzero mean. Yet, even as $\lambda$ becomes more negative—and this potentially undesirable effect becomes more pronounced—Figure\ref{ar3_shift_full} (bottom panel) shows that the fully decoupled predictor still  extends lead time for our zero-mean AR(3) process. Because the DFP criterion limits this sign reversal to specific frequencies, we accept the trend/mean inversion as the cost of gaining increased phase excess, i.e., lead time, under zero-mean assumptions. Pushing $\lambda$ too negative, however, risks violating the positivity condition, which mandates a strictly positive target correlation, see Section \ref{limit}. While our proposed predictors satisfy this condition (Fig. \ref{ar3_shift_full}, middle-right), the plot indicates that decreasing $\lambda$ past the point of full decoupling risks bordering on, or effectively violating, the assumption. Thus, in this example, the fully decoupled design is near-maximal for positivity-preserving lead time.\\
Broadly speaking, generating a lead in stationary mean-reverting processes requires key components to flip signs early to anticipate the coming reversion.  The DFP criterion facilitates this by optimally weighting these components—specifically the low-frequency ones in our example—demonstrating efficacy even in non-periodic contexts. However, in such challenging scenarios, the bounds on $\lambda$ must be considered  carefully if strict positivity is deemed relevant.\\

Enforcing accurate tracking of a nonzero mean (or linear trend)—i.e., not merely avoiding mean inversion but matching the level—would require adding constraints to the MSE-DFP problem \eqref{dp2}, thereby trading off lead (for a given target correlation) or target correlation (for a given lead).\footnote{Exact mean tracking requires $\Gamma_b(0)=\Gamma_h(0)$ (first-order); exact trend tracking additionally requires $\dot{\Gamma}_b(0)=\dot{\Gamma}_h(0)$ (second-order); see McElroy and Wildi (2016).}. Meanwhile, the unit-length DFP \eqref{dp}, which optimizes correlation, ignores the mean entirely and cannot recover a fixed nonzero level without  an ex post adjustment. \\
To handle (stationary) nonzero-mean series, we propose a two-step procedure: first maximize lead for a given target correlation (or vice versa), even if this entails mean inversion; then apply a static level correction. See Heinisch et al. (2026) for an application to multi-step GDP forecasting.\\


\begin{table}[ht]
\centering
\begin{tabular}{rrrrr}
  \hline
 & $\tau(0)$ & $\Gamma(0)$ & $\lambda$ & $\alpha_0$ \\ 
  \hline
MSE(3) & -4.01 & 3.01 &  &  \\ 
  MSE(20) & -3.96 & 3.00 &  &  \\ 
  DFP-shifted & -2.01 & 1.92 & -0.49 & 0.52 \\ 
  DFP full dec. &  & -1.07 & -0.57 & 0.00 \\ 
   \hline
\end{tabular}
\caption{Left to right: the time-shifts $\tau(\omega)$ and transferfunctions $\Gamma(\omega)$ evaluated at $\omega=0$,  followed by the hyperparameters $\lambda$ and $\alpha_0$. Top to bottom: $h=3$-step-ahead  MSE (first row),  the $20-$step-ahead MSE and then the two 3-step-ahead DFP: the $\tau$-shifted DFP and the fully decoupled DFP predictors. All predictors are normalized to unit length. Increasing the forecast horizon only slightly changes the time-shift of the MSE predictor (first and second row). For the 3-step-ahead DFP predictors, negative $\lambda$ values induce phase excess (lead). In the first column, the $\tau$-shifted DFP leads the MSE by the imposed amount $\tau=-2.01-(-4.01)=2$. The fully-decoupled DFP is phase reverting ($\Gamma(0)<0$), so the time-shift is not well-defined at $\omega=0$ (and the corresponding entry is left blank).} 
\label{ar3_shift_decoupling}
\end{table}\begin{figure}[H]\begin{center}\includegraphics[height=4in, width=5in]{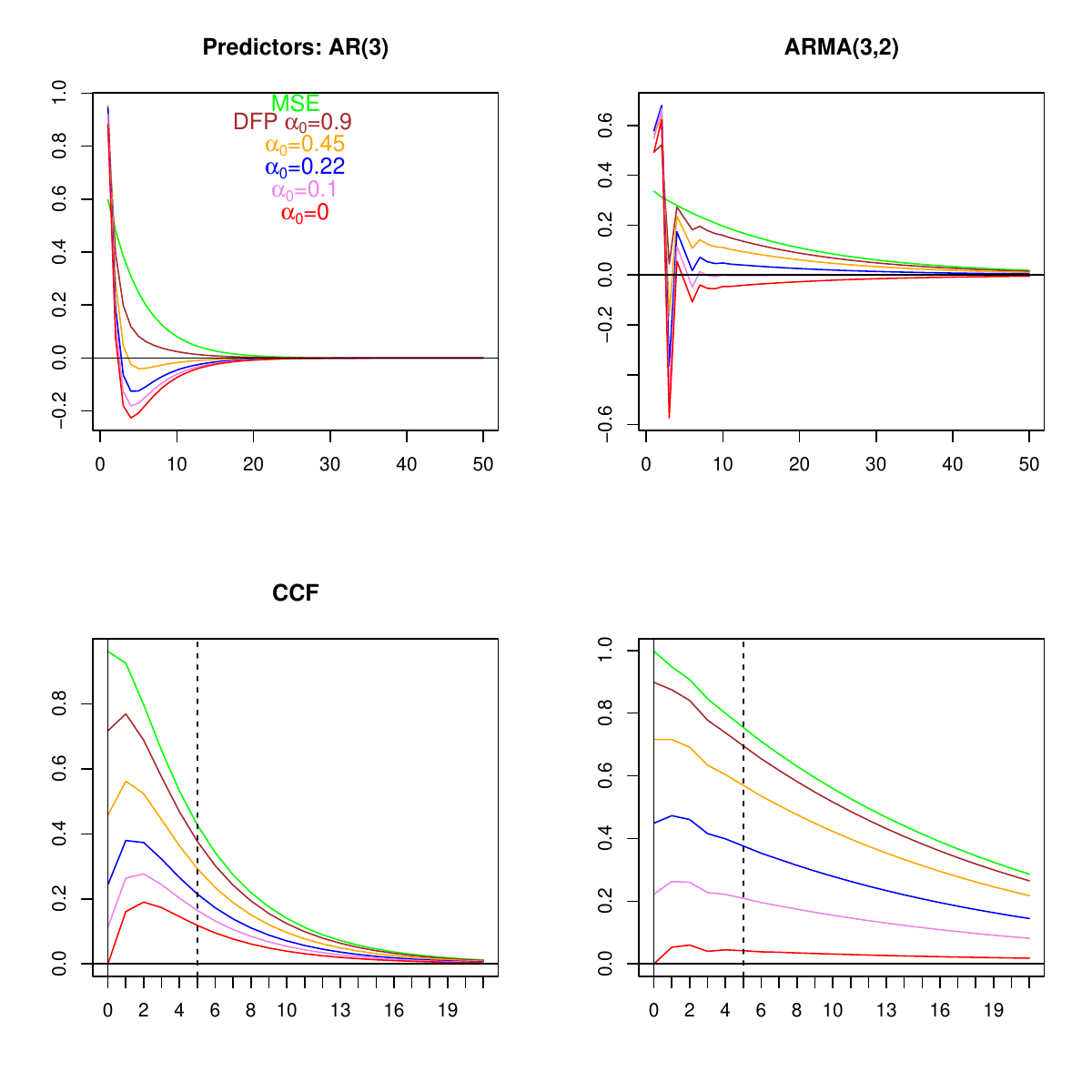}\caption{$L=50$-length moving-average inversions of the $h=5$-step-ahead predictors for the AR(3) process (left panels) and the ARMA(3,2) process (right panels) are shown: DFP predictors (top row) and the corresponding CCFs (bottom row). To facilitate visual comparison, all predictors are rescaled to have unit norm. For both processes, the DFP designs are generated using the sequence of decoupling parameter values $\alpha_0$ indicated in the upper-left panel.\label{ar3_canc}}\end{center}\end{figure}\begin{figure}[H]\begin{center}\includegraphics[height=4in, width=6in]{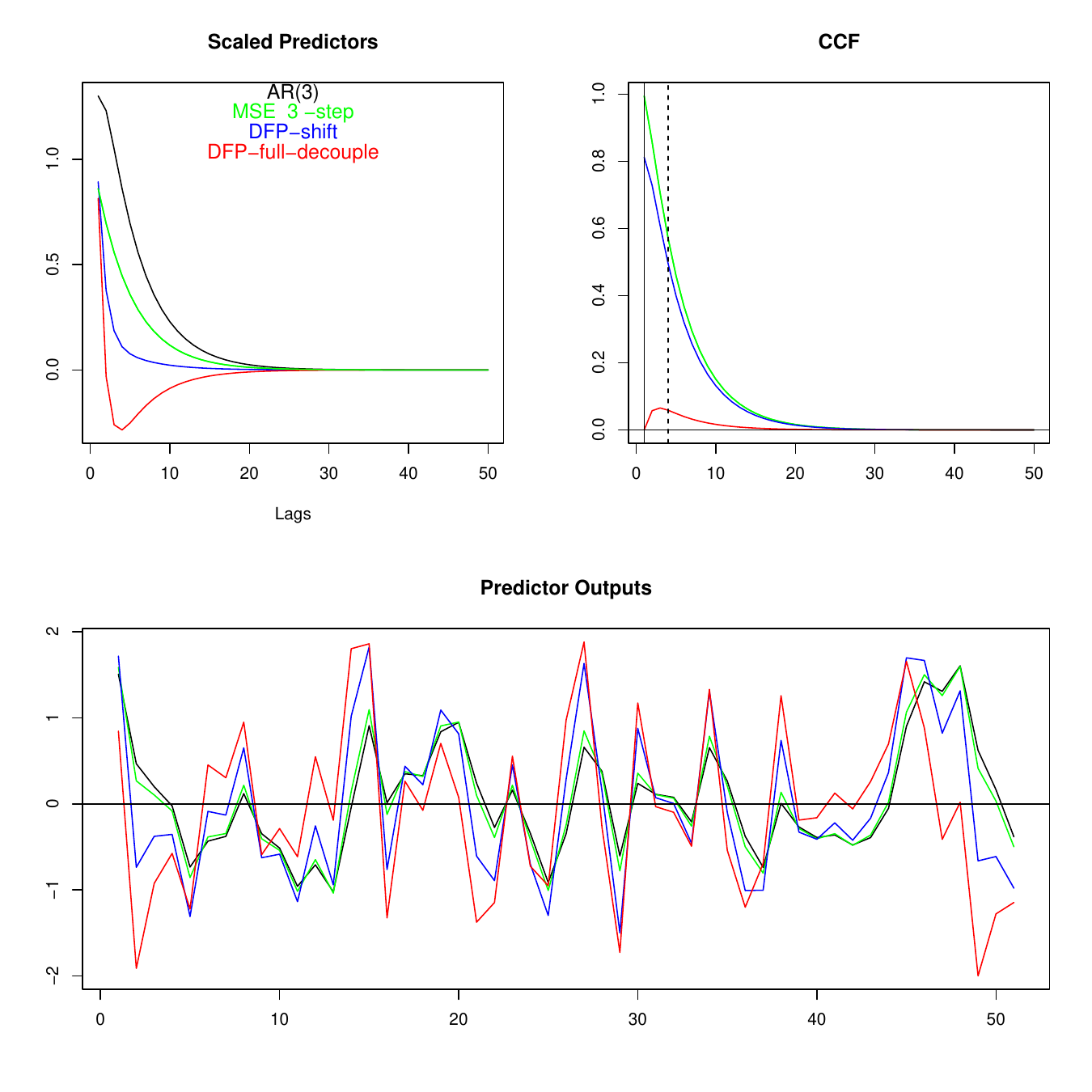}\caption{Filter coefficients (top left), CCF (top right), and filter outputs (bottom) for the AR(3) filter (black), the $h=3$-step-ahead MSE predictor (green), DFP with a time shift  $\tau=2$ (advancement, blue line), and fully decoupled DFP (red). For easier visual comparison, all filters are normalized to unit norm (equivalently, unit output variance when driven by standardized white noise). Both DFP predictors are consistently shifted to the left relative to the MSE predictor (bottom panel). \label{ar3_shift_full}}\end{center}\end{figure}

\subsection{Application of PCS to Real Time Business-Cycle Analysis}\label{RTBCA}

\begin{figure}[H]\begin{center}\includegraphics[height=4in, width=6in]{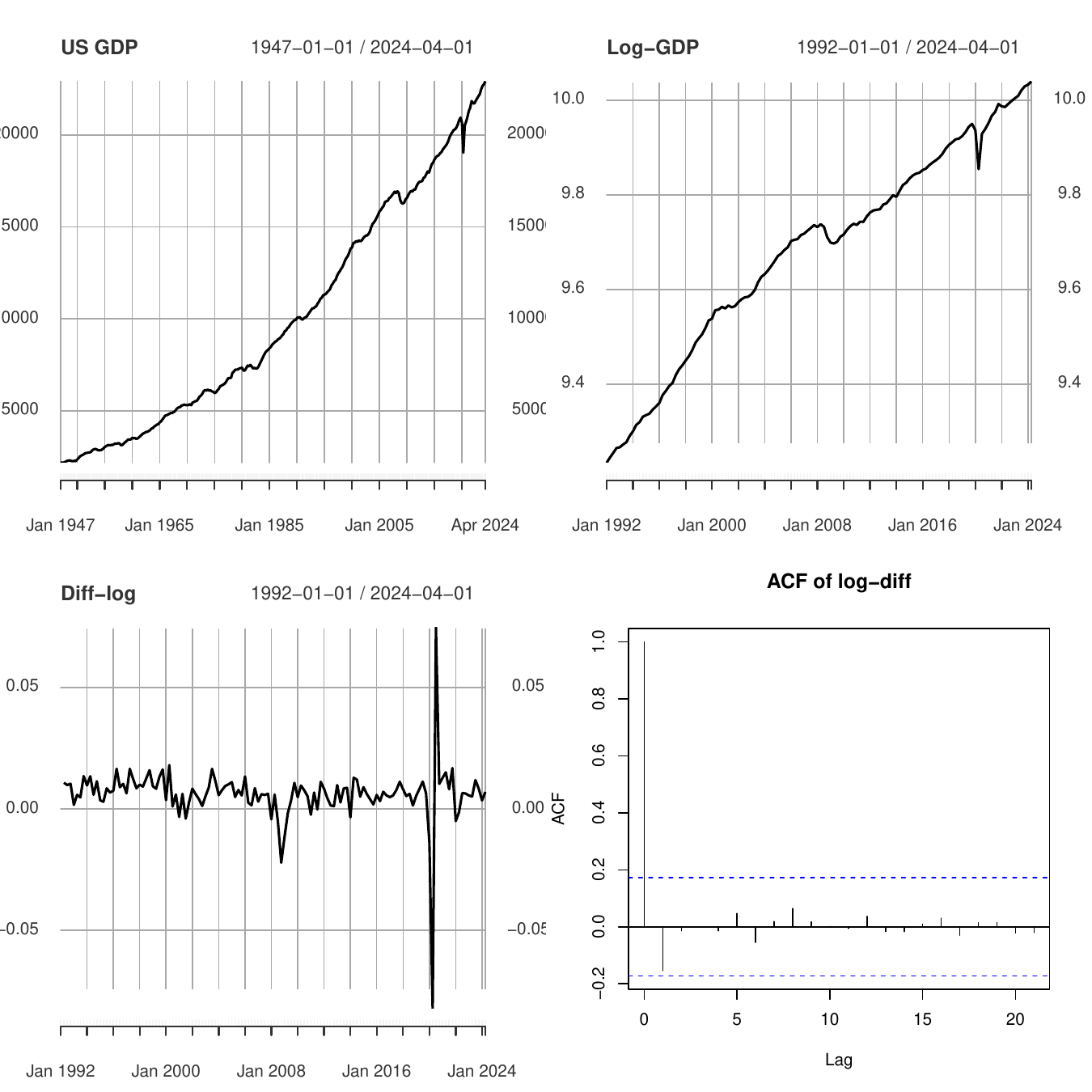}\caption{US GDP long time span (top left), log-GDP  from 1992 to 2024 (top right), log-differences (bottom left) and ACF of log-differences (bottom right).\label{GDP}}\end{center}\end{figure}

Business cycle analysis characterizes fluctuations in macroeconomic activity either as deviations from a smooth potential-growth path (classical cycle) or via movements in growth rates (growth-cycle). Both perspectives require a reliable estimate of the underlying trend. We extract the trend (growth) component of quarterly U.S. real GDP using the Hodrick–Prescott (HP) filter (Hodrick and Prescott, 1997), with the conventional quarterly smoothing parameter $\lambda_{HP} = 1600$. The data, retrieved from the FRED database (https://fred.stlouisfed.org/), are displayed in Fig. \ref{GDP}. The sample covers the last three recession episodes, spanning 1992-01-01 to 2024-04-01. \\

The HP filter is intrinsically acausal, two-sided, and bi-infinite. In real time, however, dependence on future observations is infeasible, complicating estimation. To accommodate this constraint, the HP filter admits a finite, one-sided approximation for nowcasting that tracks the two-sided design optimally under suitable data-generating assumptions (McElroy, 2008; Cornea-Madeira, 2017). Building on this framework, Wildi (2024) develops a customized concurrent (real-time) HP filter that explicitly trades off accuracy and smoothness to improve nowcast reliability. We extend this approach by applying the PCS principle to induce a predictive lead relative to the MSE benchmark.  \\

We construct a leading indicator using the PCS variant in \eqref{pcs_leading}, targeting a one-year lead ($h=4$) and specifying $\boldsymbol{\gamma}=\boldsymbol{\gamma}_0$ as the (causal, finite-length) HP(1600) trend nowcast applied to the log-differenced series (Fig. \ref{GDP}, bottom-left panel). Note that this leading- indicator PCS is not a classic forecast, since the target is the nowcast, albeit with its CCF shifted. \\
We evaluate the PCS objective over a range of hyperparameter values $\beta_h=\beta_{4}$ and compare the resulting filters with the classical one-year MSE predictor under a white-noise assumption (justified by the bottom-right panel in Fig. \ref{GDP}). All predictors are constructed with length $L=50$.\\

Figure \ref{shift_peak_coef_hp} compares the MSE predictor (green) with (leading-indicator) PCS-based predictors for several values of  $\beta_{4}$ (top left panel); all designs are normalized to unit-norm. As intended, a CCF peak shift to $\delta=4$ is achieved at $\beta_{4}=0$, with predictor 
\[
\mathbf{{b}}:=\boldsymbol{\gamma}-8.81(\mathbf{F}^{3}-\mathbf{F}^4)\boldsymbol{\gamma}
\]
(violet line), 
where $\boldsymbol{\gamma}=\boldsymbol{\gamma}_0$ is the indicator  underlying the leading-design construction (HP nowcast). \\
The cross-correlation functions (top right panel) indicate that, as $\beta_{4}$  approaches zero, the peak moves toward $\delta=h=4$ (violet line). The standardized filter outputs (bottom panel) likewise show that the predictive lead of the PCS predictors increases as $\beta_{4}$ decreases.
 
\begin{figure}[H]\begin{center}\includegraphics[height=5in, width=5in]{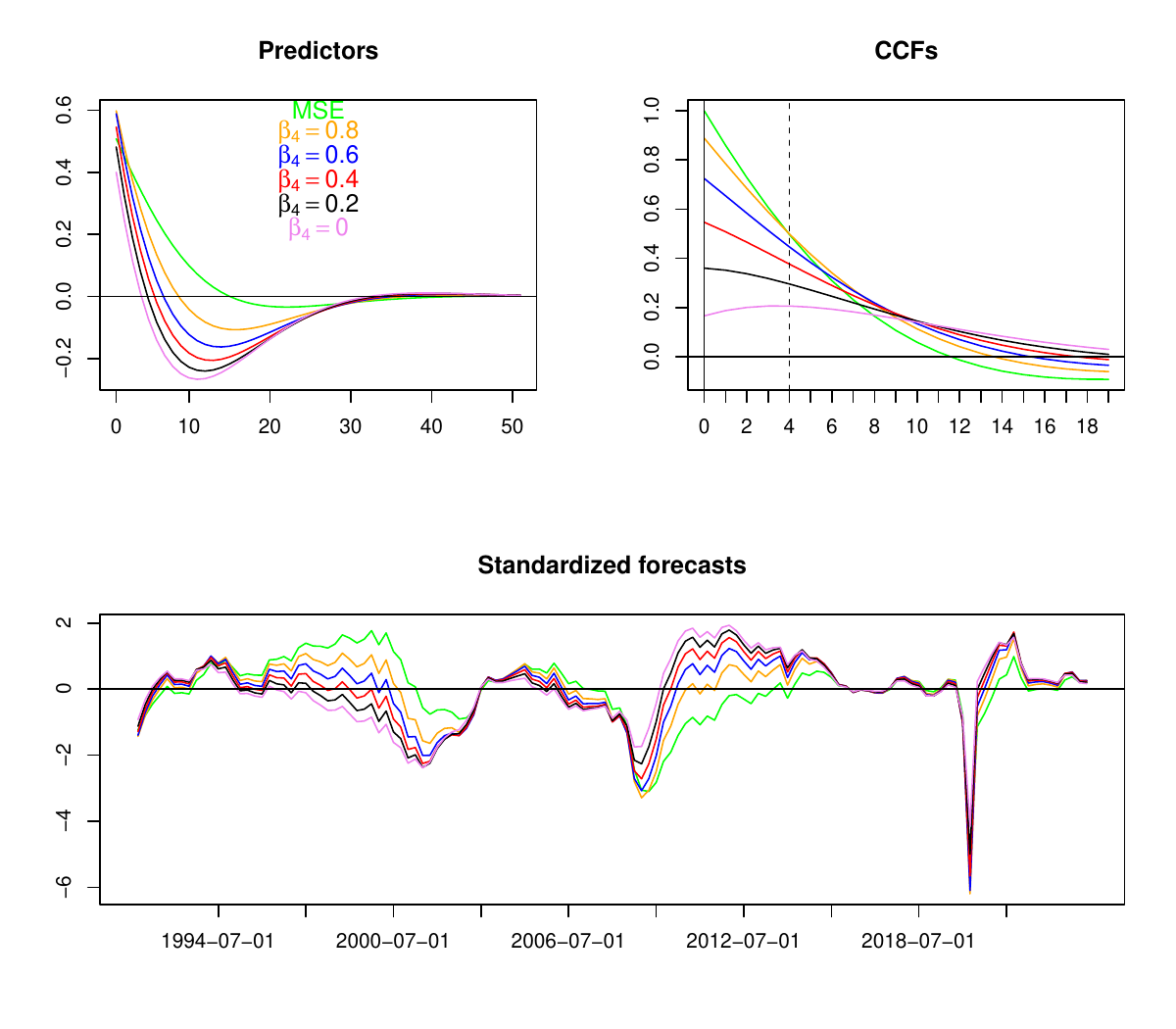}\caption{Filter coefficients (top left), CCFs (top right) and  leading indicators (bottom) of  MSE- (green) and PCS-predictors  under various choices of $\beta_4$. When $\beta_4=0$, the CCF attains its peak at the forecast horizon (violet line top right panel). Predictors and series are normalized for clearer visual inspection.\label{shift_peak_coef_hp}}\end{center}\end{figure}

\section{Conclusion}\label{concl}

We address the inherent accuracy–timeliness conflict in prediction by embedding it in an objective–constraint optimization framework. A formal definition of predictor `lead' maps this dilemma to an explicit MSE–lead trade-off.\\

We propose new look-ahead designs: i) unit-length and MSE-DFP, which decouple the forecast from the nowcast, and ii) unit-length, MSE, and leading-indicator PCS, which shift the CCF peak toward the forecast horizon. Closed-form solutions show that these designs are generally geometrically distinct, lying in different subspaces of the predictor space. We also derive their distributions, linking them to standard MSE-predictor theory through appropriate transformations. \\

For clarity, we focus on stationary univariate processes with full-rank predictor spaces, though the framework extends to nonstationary or multivariate settings and to singular (rank-deficient) spaces. We illustrate the approach in time-series forecasting and real-time signal extraction, including leading-indicator design. In each application, a single hyperparameter controls the accuracy–timeliness trade-off and admits clear statistical and geometric interpretations, tracing a new accuracy–timeliness efficient frontier. Whereas the MSE predictor occupies a single spot on this curve, our approaches span the entire frontier.\\

Future work could combine this framework with Wildi’s recent results (2024–2026) to develop a unified `prediction trilemma' that jointly balances accuracy, timeliness, and smoothness.

 \appendix

 \section{Lead in the Non-Standard Case $\tau_h(0)<\tau_0(0)$}\label{l_ns_s}

 We argue that the final baseline condition $\tau_h(0)>\tau_0(0)$ is not strictly necessary to derive an expression for the phase differential $\Phi_b(\omega)-\Phi_h(\omega)$ in  Proposition \eqref{alpha_0_beta}, although it does simplify both the analysis and the derivation. Suppose instead that this condition is violated, i.e., $\tau_h(0)<\tau_0(0)$ (non-standard situation). Using \eqref{ts_0}, we require 
\[
-\frac{\sum_{k=1}^{L-1}k\gamma_k}{\sum_{k=0}^{L-1}\gamma_k}=\tau_0(0)>\tau_h(0)=-\frac{\sum_{k=1}^{L-1}k\gamma_{k+h}}{\sum_{k=0}^{L-1}\gamma_{k+h}}.
\]
As an example, for $L=2$ and $h=1$, with $\gamma_0=1$, this reduces to  $\gamma_1^2<\gamma_2$. In that case the one-step predictor $\boldsymbol{\gamma}_1=(\gamma_1,\gamma_2)'$ lags the nowcast $\boldsymbol{\gamma}_0=(1,\gamma_1)'$ at frequency zero.\\

In this non-standard case 
\[
\Phi_h(\omega)\approx\dot{\Phi}_h(0)\omega=\tau_h(0)\omega<\tau_0(0)\omega=\dot{\Phi}_0(0)\omega\approx\Phi_0(\omega),
\]
for small $\omega>0$. 
For $\lambda<0$ (phase excess) in $\Gamma_b=\Gamma_h+\lambda\Gamma_0$ this suggests $\Phi_b(\omega)<\Phi_h(\omega)$ and hence
\[
\tau_b(0)\approx \Phi_b(\omega)/\omega<\Phi_h(\omega)/\omega\approx\tau_h(0),
\]
so the DFP would \emph{lag} the MSE predictor at the reference (zero) frequency. 
Nevertheless, a lead at zero frequency is still attainable. 
Specifically, one can distinguish two scenarios:  either i)
\begin{eqnarray}\label{phase_ord1}
\Phi_0(\omega)\geq \Phi_b(\omega)\geq\Phi_h(\omega),
\end{eqnarray}
or ii)
\begin{eqnarray}\label{phase_ord2}
\Phi_b(\omega)\geq\Phi_0(\omega)\geq\Phi_h(\omega),
\end{eqnarray}
both of which differ from \eqref{phase_ord}.\\ 
In the first scenario, $\lambda>0$ implies $\Phi_b(\omega)>\Phi_h(\omega)$\footnote{If $\lambda<0$, the DFP rotates in the same direction as the MSE predictor relative to the nowcast: in that case $\theta_{0h}<0$ and the `phase excess' $\theta_{0b}=\theta_{0h}+\theta_{hb}<\theta_{0h}<0$ increases the \emph{lag}. Rotating in the opposite direction, i.e., $\theta_{hb}>0$, to increase the \emph{lead} therefore requires $\lambda>0$ in the non-standard case.} and
\[
\tau_b(0)\approx\Phi_b(\omega)/\omega>\Phi_h(\omega)/\omega\approx\tau_h(0)
\]
inducing a lead, as desired. However, in case i) $\Phi_b(\omega)<\Phi_0(\omega)$ so that
\begin{eqnarray}\label{phib<phi0}
\tau_b(0)\approx\Phi_b(\omega)/\omega<\Phi_0(\omega)/\omega\approx\tau_0(0),
\end{eqnarray}
that is, the DFP time shift does not exceed the nowcast. For a graphical illustration, it is convenient to use Fig. \ref{dfp_geometry} for the unit-length DFP Criterion \eqref{dp}: case (i) is obtained by interchanging $\boldsymbol{\gamma}_h$ (i.e., $\Gamma_h(\omega)$) and $\boldsymbol{\gamma}_0$ (i.e.,$\Gamma_0(\omega)$), with $\mathbf{b}$ (i.e., $\Gamma_b(\omega)$) lying between them. Accordingly, the DFP solution corresponds to the \emph{lower} branch of the cone, represented in the figure by $\theta_{0b}<0$.  As $\lambda\to\infty$ (equivalently $\alpha_0\to\infty$), $\mathbf{b}(\lambda)\propto\boldsymbol{\gamma}_0$ and hence $\tau_b(0)\to\tau_0(0)$.\\
Here $\beta(\omega)=\Phi_b(\omega)-\Phi_h(\omega)$ is determined from the triangle with sides $a,b,c$ as in Proposition \eqref{alpha_0_beta} but with
\[
\alpha(\omega)=\Phi_0(\omega)-\Phi_h(\omega)>0
\]
(instead of $\gamma(\omega)=\Phi_h(\omega)-\Phi_0(\omega)>0$). By the  Law of Sines
\[
\frac{\sin(\beta(\omega))}{b(\omega)}=\frac{\sin(\alpha(\omega))}{a(\omega)},
\]
so $\beta$ follows from this relation, with $a=A_h,b=\lambda A_0$ for $\lambda>0$.\\
In case ii), one instead selects the upper branch of the cone, represented by $\theta_{0b}>0$ in Fig.\ref{dfp_geometry}. In this non-standard setting, recovering the upper-branch solution from the unit-length criterion \eqref{dp} requires switching the maximization of the target correlation to a \emph{minimization}, with $\lambda>0$ again ensuring a lead. This implies
\[
\mathbf{b}=\lambda_1\boldsymbol{\gamma}_h+\lambda_2\boldsymbol{\gamma}_0
\]
with the weight  on $\boldsymbol{\gamma}_h$ \emph{negative}, $\lambda_1<0$. In principle, the same sign reversal carries over to the MSE-DFP; however, swapping sign, i.e., maximizing MSE in  \eqref{dp2}, fails because $\mathbf{b}$ is not constrained to unit-length. An alternative is to swap the roles of $\boldsymbol{\gamma}_0$ (now in the objective) and $\boldsymbol{\gamma}_h$ (now in the constraint)
\begin{eqnarray}
&&\min_{\mathbf{b}} (\mathbf{b}-\boldsymbol{\gamma}_0)'(\mathbf{b}-\boldsymbol{\gamma}_0)\label{dp3}\\
\textrm{s.t.}&&\boldsymbol{\gamma}_{h}'\mathbf{b}=\alpha_h.\nonumber 
\end{eqnarray}
Then Proposition \ref{dfp_mse_prop} yields
\begin{eqnarray}\label{caseii}
\mathbf{b}(\lambda)=\boldsymbol{\gamma}_0+\lambda\boldsymbol{\gamma}_h.
\end{eqnarray}
For $\lambda<0$, $\tau_{b}(0)>\tau_0(0)$, as desired. \\
To derive $\beta(\omega)=\Phi_b(\omega)-\Phi_h(\omega)$ in this  case,  Proposition \ref{alpha_0_beta} continues to apply with these changes: $a=A_0$ (instead of $A_h$),  $b=|\lambda| A_h$ (instead of $|\lambda|A_0$) and flip the sign of $\gamma=\Phi_0-\Phi_h$. \\

Constructing `large' leads  over the MSE predictor  under the non-standard ordering $\tau_h(0)<\tau_0(0)$ (case ii)) requires either flipping the sign in criterion \eqref{dp} (i.e., minimizing instead of maximizing) or swapping the roles of the nowcast and the MSE predictor in criterion \eqref{dp3}, thereby extending the standard DFP setup.\\

This analysis demonstrates that the (last) baseline condition $\tau_h(0)>\tau_0(0)$ (the standard case),  is not essential to derive $\beta(\omega)$ in Proposition \eqref{alpha_0_beta}; it mainly streamlines the exposition by preserving, near the reference zero frequency, the DFP-triangle geometry implied by \eqref{phase_ord}. In the standard case, rotating the DFP predictor in the same direction as the MSE predictor relative to the nowcast (i.e., $\lambda<0$ and minimizing MSE) produces a lead $\tau_b(0)>\tau_h(0)$ at the reference frequency. This matches the previous section’s DFP geometry: the phase excess $\theta_{0b}-\theta_{0h}>0$ implied by $\lambda<0$ corresponds to a lead  at frequency zero. When $\tau_h(0)<\tau_0(0)$,  the extensions based on the reorderings in \eqref{phase_ord1} and \eqref{phase_ord2} apply. Then $\mathbf{b}=\boldsymbol{\gamma}_h+\lambda\boldsymbol{\gamma}_0$ with $\lambda>0$ (case i)) or $\mathbf{b}=\boldsymbol{\gamma}_0+\lambda\boldsymbol{\gamma}_h$ with $\lambda<0$ (case ii)).\\ 

 \section{Linking Phase Excess and Hyperparameter in the DFP Criterion}\label{dfp_link}

 Let $\theta_{0b}$ and $\theta_{0h}$ denote the angles between $\boldsymbol{\gamma}_{0}$ and $\mathbf{{b}}$, and between $\boldsymbol{\gamma}_{0}$  and $\boldsymbol{\gamma}_{h}$, respectively.  The next proposition relates $\lambda$ to these angles when $\theta_{0b}>\theta_{0h}>0$ (phase excess).

\begin{Proposition}\label{proplambda0}
Assume $\boldsymbol{\gamma}_0\not\propto\boldsymbol{\gamma}_h$ (rank-two) and $\theta_{0b}>\theta_{0h}>0$ (phase excess). Then  
\begin{eqnarray}\label{lambda0} 
\lambda=-\frac{|\boldsymbol{\gamma}_h|}{|\boldsymbol{\gamma}_0|}\frac{\sin(\theta_{0b}-\theta_{0h})}{\sin(\theta_{0b})}
\approx\frac{|\boldsymbol{\gamma}_h|}{|\boldsymbol{\gamma}_0|}\left(\frac{\theta_{0h}}{\theta_{0b}}-1\right),
\end{eqnarray}
where the approximation is valid  for sufficiently small $\theta_{0b}$.
\end{Proposition}

\textbf{Proof}: Geometrically, the vectors $\mathbf{b}$, $\boldsymbol{\gamma}_{0}$ and $\boldsymbol{\gamma}_{h}$ lie in the same plane and define the DFP-triangle with side lengths   $a=|\boldsymbol{\gamma}_h|$, $b=|\lambda||\boldsymbol{\gamma}_0|$, and $c=|\mathbf{b}|$; see Fig.\ref{dfp_geometry_2}. When $\theta_{0b}>\theta_{0h}>0$, the corresponding opposite angles are $\beta=\theta_{0b}-\theta_{0h}$, $\gamma=\theta_{0h}$ and $\alpha=\pi-(\theta_{0b}-\theta_{0h})-\theta_{0h}=\pi-\theta_{0b}$.\footnote{If, instead, $0<\theta_{0b}<\theta_{0h}$ (phase loss), the opposite angles become $\beta=\theta_{0h}-\theta_{0b}$, $\gamma=\pi-\theta_{0h}$ and $\alpha=\theta_{0b}$.} Note that  $\theta_{0b}-\theta_{0h}<\pi$ by construction of $\mathbf{b}$.  Using the law of sines gives  $\frac{\sin(\beta)}{b}=\frac{\sin(\alpha)}{a}$, and hence 
\[
|{\lambda}|=\frac{|\boldsymbol{\gamma}_h|\sin(\theta_{0b}-\theta_{0h})}{|\boldsymbol{\gamma}_0|\sin(\pi-\theta_{0b})}.
\]
In addition,  the triangle construction together with $\theta_{0b}>\theta_{0h}>0$ implies ${\lambda}<0$. Therefore, \begin{eqnarray}\label{gtre}
0>{\lambda}=-\frac{|\boldsymbol{\gamma}_h|}{|\boldsymbol{\gamma}_0|}\frac{\sin(\theta_{0b}-\theta_{0h})}{\sin(\theta_{0b})}\approx \frac{|\boldsymbol{\gamma}_h|}{|\boldsymbol{\gamma}_0|}\left(\frac{\theta_{0h}}{\theta_{0b}}-1\right),
\end{eqnarray}
where the approximation follows from a first-order Taylor expansion of the sine function when $\theta_{0b}(>\theta_{0h}>0)$ is small. \\

\textbf{Remarks}: If $\theta_{0b}>\theta_{0h}>0$ (phase excess), then $\lambda<0$, and the proposed look-ahead perspective highlights differences rather than similarities between the MSE predictor $\boldsymbol{\gamma}_h$ and the nowcast $\boldsymbol{\gamma}_0$; see Section \ref{AR3} for illustration. In contrast, if  $\theta_{0h}>\theta_{0b}>0$ (phase loss), then $\lambda>0$ and the resulting DFP design stresses their commonalities, albeit with a corresponding lag relative to the MSE predictor.  Finally, flipping the signs of $\theta_{0h}$ and $\theta_{0b}$ so that $0>\theta_{0h}>\theta_{0b}$ leaves these conclusions unchanged, since the signs cancel in Equation \eqref{lambda0}. Geometrically, both orderings require the DFP predictor to rotate farther in the same direction than the MSE predictor within the plane spanned by  $\boldsymbol{\gamma}_0,\boldsymbol{\gamma}_h$.\\

\begin{Corollary}\label{cor_al}
The hyperparameter $\alpha_0$ of the DFP-MSE Criterion \eqref{dp2} can be linked to the phase excess $\theta_{0b}$, assuming $\theta_{0b}>\theta_{0h}>0$, viz.
\[
\alpha_0=|\boldsymbol{\gamma}_h||\boldsymbol{\gamma}_0|\left(\cos(\theta_{0b})-\frac{\sin(\theta_{0b}-\theta_{0h})}{\sin(\theta_{0b})}\right).
\]
\end{Corollary}
The proof follows from inserting \eqref{lambda_mse} into \eqref{lambda0} and solving for $\alpha_0$.

\section{Linking Phase Excess and Hyperparameter in the PCS Criterion}\label{phe_bet}

The next proposition establishes the relationship between the phase excess $\theta_{hb}$ and the PCS hyperparameter $\beta_h$ in \eqref{peak_cor_sh_crit}. Throughout, we assume that the constraint set is feasible and that the resulting PCS solution attains a strictly positive objective value.

\begin{Proposition}\label{l_t_b}
Let $\boldsymbol{\gamma}_{h-1}$ and $\boldsymbol{\gamma}_h$ be linearly independent with  $\theta_{hh-1}>0$. Define 
\begin{eqnarray*}
a&:=&\|\boldsymbol{\gamma}_h\|-\cos(\theta_{hh1})\|\boldsymbol{\gamma}_{h-1}\|\\
b&:=&\sin(\theta_{hh-1})\|\boldsymbol{\gamma}_{h-1}\|\\
R&:=&\sqrt{a^2+b^2}\\
\phi&:=&\arctantwo(b,a)
\end{eqnarray*}
and assume that $R>0$ and that $|\beta_h|\leq R$. Let $\theta_{hb}>0$  (phase excess) denote the angle between $\boldsymbol{\gamma}_h$ and the PCS solution $\mathbf{b}$ within the plane spanned by $\boldsymbol{\gamma}_{h-1}$ and $\boldsymbol{\gamma}_h$ (see Fig.\ref{pcs_geometry}).  If $|\beta_h|=R$, then 
\[
\theta_{hb}=\left\{\begin{array}{cc}\phi&,\beta_h=-R\\
\phi-\pi&,\beta_h=R~\textrm{and~}0<\phi<\pi\\
\phi+\pi&,\beta_h=R~\textrm{and~}\pi<\phi<2\pi.
\end{array}\right.
\]
If $|\beta_h|<R$ then 
\[
\theta_{hb}=\left\{\begin{array}{cc}\phi-\arccos(-\beta_h/R)&,0< \phi<\pi\\
\phi+\arccos(-\beta_h/R)&,\pi< \phi<2\pi.\end{array}\right.
\]
\end{Proposition}

\textrm{Proof}: We first consider the case $|\beta_h|<R$. From the geometry  (blue case, $\beta_h<0$ in Fig. \ref{pcs_geometry}) 
\[
\beta_h=\cos(\theta_{hb}+\theta_{hh-1})\|\boldsymbol{\gamma}_{h-1}\|-\cos(\theta_{hb})\|\boldsymbol{\gamma}_h\|=-(a\cos(\theta_{hb})+b\sin(\theta_{hb})),
\]
with $a=\|\boldsymbol{\gamma}_h\|-\cos(\theta_{hh-1})\|\boldsymbol{\gamma}_{h-1}\|$, $b=\sin(\theta_{hh-1})\|\boldsymbol{\gamma}_{h-1}\|$. Let $R=\sqrt{a^2+b^2}$, 
and consider the right triangle with legs $a$ and $b$ and hypotenuse $R$, assuming $a, b > 0$. For general sign combinations, the use of $\phi=\arctantwo(b, a)$\footnote{The $\arctantwo$ function returns the angle between the x-axis and the vector from the origin to $(x,y)$, i.e., for positive arguments $\arctantwo(y, x) = \arctan(y/x)$.} accounts for possibly negative values of $a$ or $b$. Let $\phi$ designate the angle adjacent to $a$ and $R$ so that $\cos(\phi)=a/R$ and $\sin(\phi)=b/R$. Therefore, using the identity 
\[
a\cos(\theta_{hb})+b\sin(\theta_{hb})=R\cos(\theta_{hb}-\phi),
\]
we obtain $\beta_h=-R\cos(\theta_{hb}-\phi)$. Inverting gives 
\[
\theta_{hb}^{\pm}=\phi\pm\arccos(-\beta_h/R).
\]
Among the two candidates, choose the branch that maximizes the objective  $\mathbf{b}(\theta_{hb}^{\pm})'\boldsymbol{\gamma}_h=\cos(\theta_{hb}^{\pm})\|\boldsymbol{\gamma}_h\|$. Using
\begin{eqnarray*}
\cos(\theta_{hb}^-)-\cos(\theta_{hb}^+)&=&\cos\Big(\phi-\arccos(-\beta_h/R)\Big)-\cos\Big(\phi+\arccos(-\beta_h/R)\Big)\\
&=&2\sin(\phi)\sin\Big(\arccos(-\beta_h/R)\Big)=2\sin(\phi)\sqrt{1-\beta_h^2/R^2},
\end{eqnarray*}
we conclude that $\cos\Big(\theta_{hb}^-\Big)>\cos\Big(\theta_{hb}^+\Big)$ when $\phi\in ]0,\pi[$, and the inequality reverses when $\phi\in]\pi,2\pi[$. The cases $\phi=0\pmod{2\pi}$ and $\phi=\pi\pmod{2\pi}$ are excluded, as they would imply $b = 0$ (i.e., $\sin(\theta_{hh-1})=0$), which in turn yields collinear MSE predictors. Therefore, maximization of the objective implies that
\[
\theta_{hb}=\left\{\begin{array}{cc}\phi-\arccos(-\beta_h/R)&,0< \phi<\pi\\
\phi+\arccos(-\beta_h/R)&,\pi< \phi<2\pi.\end{array}\right.
\]
The solution when $|\beta_h|=R$ follows as a special case of the above.\hfill\qed\\

\textbf{Remarks:} When $\theta_{hb}=0$, the PCS predictor coincides with the $h$-step MSE predictor $\boldsymbol{\gamma}_h$. In this case, $\mathbf{b}'\boldsymbol{\gamma}_h=\|\boldsymbol{\gamma}_h\|$ and $\mathbf{b}'\boldsymbol{\gamma}_{h-1}=\|\boldsymbol{\gamma}_{h-1}\|\cos(\theta_{hh-1})$. Therefore, the PCS constraint parameter that reproduces the MSE predictor is $\beta_h^{MSE}=\|\boldsymbol{\gamma}_{h-1}\|\cos(\theta_{hh-1})-\|\boldsymbol{\gamma}_h\|$. For $\beta_h<\beta_h^{MSE}$, the phase excess $\theta_{hb}$ becomes positive, indicating an advancement of the PCS relative to the MSE predictor. Note that the positivity of the phases $\theta_{hh-1}>0$ and $\theta_{b}>0$ 
 is purely conventional. The essential requirement is that both quantities share the same sign, i.e., $\theta_{hh-1}\theta_{b}>0$. Equivalently, this condition ensures that $\mathbf{b}$ is positioned on the side of $\boldsymbol{\gamma}_h$ opposite to $\boldsymbol{\gamma}_{h-1}$. Under this alignment, the rotation from $\boldsymbol{\gamma}_{h}$ to $\mathbf{b}$ by angle $\theta_{hb}$ proceeds in the same direction as the rotation from $\boldsymbol{\gamma}_{h-1}$ to $\boldsymbol{\gamma}_{h}$: if $\boldsymbol{\gamma}_{h}$ constitutes an advancement relative to $\boldsymbol{\gamma}_{h-1}$, then $\mathbf{b}$ likewise represents an advancement relative to $\boldsymbol{\gamma}_{h}$. The magnitude of this advancement is governed by the hyperparameter $\beta_h$, with $\beta_h<\beta_h^{MSE}$ ensuring a positive phase excess.  From an interpretability standpoint, the PCS geometry centers on displacing the peak correlation, yielding an intuitively appealing measure of advancement (see De Jong and Nijman, 1997). A link between the phase excess $\theta_{hb}$ and the time-shift $\tau_b(0)$ at frequency $\omega=0$ can be derived from Section \ref{time_shift} (details omitted).

 \section{Extension of the PCS Criterion}\label{pcs_ext}

We extend the PCS approach by recognizing that, in some applications, imposing the PCS constraint only at the single lead $\delta = h$—i.e., constraining $(\boldsymbol{\gamma}_{h-1}-\boldsymbol{\gamma}_{h})'\mathbf{b}$ in \eqref{peak_cor_sh_crit}—may be insufficient to produce an effective peak shift. A stronger formulation can be obtained by imposing constraints over a neighborhood of leads, which more firmly anchors the desired peak displacement. Specifically, we propose
\begin{eqnarray*}
&&\max_{\mathbf{b}} \boldsymbol{\gamma}_{{h}}'\mathbf{b}\\
&&(\boldsymbol{\gamma}_{h-\delta}-\boldsymbol{\gamma}_{h-(\delta-1)})'\mathbf{b}=\beta_{\delta}~,~\delta\in \Delta\nonumber\\ 
&&\mathbf{b}'\mathbf{b}=1,\nonumber
\end{eqnarray*}
where $\Delta$ is a set of lead times (e.g., $\Delta=\{1,...,h\}$) and where the constraints are assumed to be feasible (unfeasibility is discussed below). Choosing $\beta_{\delta}< 0$ promotes a monotonically increasing CCF over $\delta\in\Delta$. However, the associated geometry—given by the intersection of $|\Delta|$-many cones with the unit sphere and a plane determined by the objective (see the proof of Theorem \ref{solution_dfp})—can be intricate and yield multiple solutions. For this reason, we adopt the mean-square reformulation:
\begin{eqnarray}
&&\min_{\mathbf{b}} (\boldsymbol{\gamma}_{{h}}-\mathbf{b})'(\boldsymbol{\gamma}_{{h}}-\mathbf{b})\label{peak_cor_sh_crit2}\\
&&(\boldsymbol{\gamma}_{h-\delta}-\boldsymbol{\gamma}_{h-(\delta-1)})'\mathbf{b}=\beta_{\delta}~,~\delta\in \Delta\nonumber
\end{eqnarray}
which removes the unit‑length constraint and allows the restrictions to be viewed as affine hyperplanes rather than cones. This linearized formulation yields a unique solution (under the stated regularity conditions), as shown in Proposition \ref{dfp_mse_prop}.  \\

A potential limitation of the extended (MSE-)PCS criterion \eqref{peak_cor_sh_crit2} is the requirement to pre-specify a set of hyperparameters $\beta_{\delta}$ for $\delta \in \Delta$. This design choice may lead to overly restrictive constraints being imposed on the system.To mitigate this, we propose an unconstrained penalized variant that addresses the constraints in aggregate:     
\begin{eqnarray}\label{eff_pcs}
&&\min_{\mathbf{b}}\left( (\boldsymbol{\gamma}_{{h}}-\mathbf{b})'(\boldsymbol{\gamma}_{{h}}-\mathbf{b})+\nu\sum_{\delta\in\Delta}\Big((\boldsymbol{\gamma}_{\delta-1}-\boldsymbol{\gamma}_{\delta})'\mathbf{b}-\beta\Big)^2\right),\label{peak_cor_sh_crit3}
\end{eqnarray}
where $\nu>0$ controls the strength of the penalty. If the individual constraints in the penalty are jointly feasible, then as $\nu\to\infty$, the aggregate penalty forces $\sum_{\delta\in\Delta}\Big((\boldsymbol{\gamma}_{\delta-1}-\boldsymbol{\gamma}_{\delta})'\mathbf{b}-\beta\Big)^2\to 0$, implying that $(\boldsymbol{\gamma}_{\delta-1}-\boldsymbol{\gamma}_{\delta})'\mathbf{b}\to\beta$ for all $\delta\in\Delta$ (the case of an unfeasible singular system is discussed below). For finite (positive) $\nu$, the aggregate penalty provides greater flexibility, so that difficult or misspecified individual constraints exert less adverse influence on forecast performance (the objective value).\\

We now briefly discuss the solutions to the proposed PCS criteria. As discussed, \eqref{peak_cor_sh_crit} is similar to the DFP formulation \eqref{dp} and a solution can be obtained from Theorem \ref{solution_dfp} upon the substitutions
\[
\boldsymbol{\gamma}_0\to\boldsymbol{\gamma}_{h-1} - \boldsymbol{\gamma}_h, \textrm{~~} \alpha_0\to \beta_h/\|\boldsymbol{\gamma}_{h-1}-\boldsymbol{\gamma}_h\|,
\]
assuming $\boldsymbol{\gamma}_{h-1}$ and $\boldsymbol{\gamma}_h$ are linearly independent. \\
Analogously, the extended PCS criterion \eqref{peak_cor_sh_crit2} is closely related to the mean-square formulation \eqref{dp2}. In particular, by the same argument as in Proposition \ref{dfp_mse_prop}, the solution has the form 
\begin{eqnarray}\label{epcs}
\mathbf{b}=\boldsymbol{\gamma}_h+\sum_{\delta\in \Delta}\lambda_{\delta}(\boldsymbol{\gamma}_{h-\delta}-\boldsymbol{\gamma}_{h-(\delta-1)}),
\end{eqnarray}
where the coefficients $\lambda_{\delta}$ solve the linear system obtained by imposing the constraints
\[
(\boldsymbol{\gamma}_{h-\delta}-\boldsymbol{\gamma}_{h-(\delta-1)})'\mathbf{b}=\beta_{\delta}~,~\delta\in \Delta.
\]
Specifically, let $\mathbf{A}$ be the matrix with rows $\mathbf{a}_{\delta}' = (\boldsymbol{\gamma}_{h-\delta}-\boldsymbol{\gamma}_{h-(\delta-1)})'$, let $\boldsymbol{\beta} = (\beta_{\delta})_{\delta\in\Delta}$ and denote $\boldsymbol{\lambda}=(\lambda_{\delta})_{\delta\in\Delta}$. Then 
\begin{equation}\label{lambda_A}
\boldsymbol{\lambda}=(\mathbf{AA}')^{-1}(\boldsymbol{\beta}-\mathbf{A}\boldsymbol{\gamma}_h)
\end{equation}
when $\mathbf{A}$ has full row rank; otherwise remove redundant (linearly dependent) constraints. Alternatively, infinitesimal perturbations of the predictors (cf. Section \ref{dfp}) can be employed to restore rank.\\
Finally, for fixed $\nu\geq 0$, the penalized extended PCS objective \eqref{peak_cor_sh_crit3} is a strictly convex quadratic function of $\mathbf{b}$ and therefore admits a unique closed-form minimizer
\begin{equation}\label{b_M}
\mathbf{b}=\mathbf{M}^{-1}\left(\boldsymbol{\gamma}_h+\nu\beta\sum_{\delta\in\Delta}\mathbf{d}_{\delta}\right),
\end{equation}
where $\mathbf{M}=\Big(\mathbf{I}+\nu\sum_{\delta\in\Delta}\mathbf{d}_{\delta}\mathbf{d}_{\delta}'\Big)$, $\mathbf{I}$ is the $L\times L$ dimensional identity, and $\mathbf{d}_{\delta}:=\boldsymbol{\gamma}_{\delta-1}-\boldsymbol{\gamma}_{\delta}$. Since $\mathbf{M}$ is strictly positive definite for any $\nu\geq 0$, it is invertible, ensuring existence and uniqueness of the solution. For finite $\nu$, the penalty allows controlled deviations from the individual hard constraints, providing flexibility to tune the predictor (e.g., to trade off phase advancement against tracking accuracy). A potential drawback of this flexibility is the interaction between the two hyperparameters, $\nu$ and $\beta$, which can induce offsetting effects and thereby obscure identification of the accuracy–timeliness trade-off. In particular, choosing a more negative $\beta$ ($<0$)—which enforces a steeper CCF slope and a larger peak shift—can be counteracted, at least to some extent, by selecting a smaller $\nu$, which reduces the penalty on constraint deviations.  \\

\textbf{Remark}: The constraint system underlying the extended PCS criterion \eqref{peak_cor_sh_crit2} need not be feasible; in particular, the matrix $\mathbf{AA}'$ in \eqref{lambda_A} can be rank-deficient. For example, if the data-generating process is AR($p$), the MA-inversions (impulse responses) $\boldsymbol{\gamma}_{\delta}$ at various $\delta\in\Delta$ are linked by the Yule–Walker equations, so the rank of the constraint system should not exceed $p$ (allowing for non-trivial maximization of the objective requires the rank to not exceed $p-1$). In effect, the entire CCF is determined by a finite subset of constraints, and specifying additional constraints may render the system mutually inconsistent. Analogously, the penalized criterion \eqref{eff_pcs} can be rank-deficient in the sense that the penalty term $\sum_{\delta\in\Delta}\left((\boldsymbol{\gamma}_{\delta-1}-\boldsymbol{\gamma}_{\delta})'\mathbf{b}-\beta\right)^2$ does not vanish as $\nu\to\infty$ when the constraints are inconsistent. Nevertheless, in this penalized case the matrix $\mathbf{M}$ in \eqref{b_M} remains full rank for $\nu\geq 0$, ensuring that the minimizer $\mathbf{b}$  exists and is unique even in the case of an unfeasible system.

\section{Simulation Study}\label{sim_stu}

We conduct a simulation study on the AR(3) process of Section \ref{AR3}  with series length $T=1000$, chosen to emphasize the asymptotic behavior described in Corollary \ref{mse_dfp_dist}. After generating 1000 series, we estimate the AR parameters via maximum likelihood. We then derive the estimated DFP predictors at horizon $h = 5$ (similar results are obtained for alternative $h$) by enforcing complete decoupling ($\alpha_0=0$) and applying MA inversion to the empirical AR equations (Proposition \ref{dfp_mse_prop}).\\

Figure \ref{sim_ar3} displays the estimates, sample means, confidence intervals, and true values for the initial 20 lags. Confidence intervals for the DFP predictors (middle left) utilize Corollary \ref{mse_dfp_dist} with $\gamma_0=1$ fixed. These approximations assume $L$ and $T$ are sufficiently large to satisfy $\boldsymbol{\gamma}_h\approx\mathbf{F}^h\boldsymbol{\gamma}_0$ and justify first-order Taylor approximations, respectively. The DFP-intervals effectively capture the reduced variance in sample estimates at lag 3. \\

Individual realizations of the DFP predictor (middle left panel) display stochastic variations stemming from the empirical impulse responses (top right panel). However, the averaged trajectory (red) aligns almost perfectly with the population values (cyan): due to overlap only the latter can be distinguished. A single-run comparison (bottom panel), truncated for visibility, confirms the DFP predictor consistently leads the MSE benchmark.

\begin{figure}[H]\begin{center}\includegraphics[height=4in, width=5in]{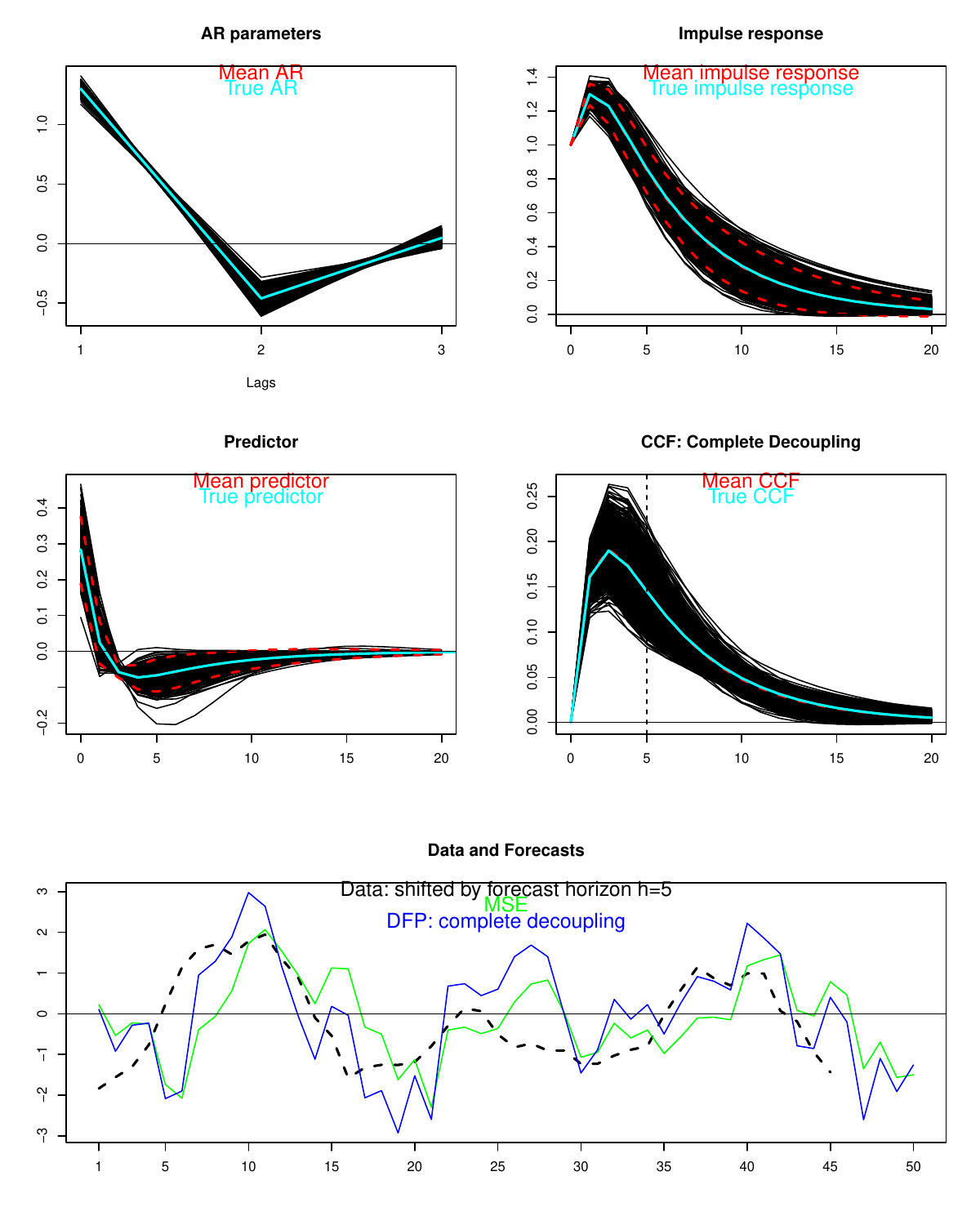}\caption{Simulation results for the MSE-DFP predictor with complete decoupling in an AR(3) process. Top-left: estimated AR coefficients; top-right: impulse response; middle-left: DFP predictor; middle-right: CCF (all CCFs are zero at lag 0). Each panel displays sample paths (black), Monte Carlo means (red) and population values (cyan); the Monte Carlo curves are largely hidden by the population curves because the biases vanish. In addition,  95$\%$  
confidence bands (red) are shown for the impulse responses and the DFP predictors. The bottom panel displays one AR(3) trajectory advanced by $h=5$ (black), together with five-step-ahead forecasts from the MSE predictor (green) and the DFP approach (blue). All series are scaled to unit variance for easier visual comparison. The DFP forecast is consistently left-shifted relative to the MSE forecast. \label{sim_ar3}}\end{center}\end{figure}

\end{document}